%% file: main.tex
\documentclass[11pt]{article}

\usepackage{cubicmoment}
\theoremstyle{remark}
\newtheorem{openproblem}[theorem]{Open Problem}
\theoremstyle{plain}
\newcommand{\gatename}[1]{\ensuremath{\text{\rmfamily\upshape #1}}}
\renewcommand{\CCZ}{\gatename{CCZ}}
\renewcommand{\CS}{\gatename{CS}}
\renewcommand{\CZ}{\gatename{CZ}}
\newcommand{\CNOT}{\gatename{CNOT}}

\usepackage{booktabs}
\usepackage{array}
\usepackage{longtable}
\newcolumntype{W}[1]{>{\raggedright\arraybackslash}p{#1}}
\usepackage{tikz}
\usetikzlibrary{arrows.meta,positioning,shapes.geometric}
\usepackage{microtype}
\usepackage{enumitem}
\setlist[enumerate]{leftmargin=*,itemsep=3pt,topsep=4pt}
\setlist[itemize]{leftmargin=*,itemsep=3pt,topsep=4pt}

\usepackage[margin=1in]{geometry}
\usepackage[colorlinks=true,linkcolor=black!70!blue,citecolor=black!60!green,
            urlcolor=black!70!blue,linktoc=all]{hyperref}

\begin{document}

\title{\textbf{Exact $T$-counts of $\CCZ$ layers from an isotropy bound}}
\author{Arul Rhik Mazumder\\ \small Imperial College London\\
        \small\texttt{rhik.mazumder26@imperial.ac.uk}}
\date{}
\maketitle

\begin{abstract}
The $T$-count is a dominant cost of fault-tolerant Clifford$+T$ computation. We prove that every
Hadamard-free Clifford$+T$ circuit with clean ancillas for a layer of $m$ disjoint $\CCZ$ gates, the
diagonal core of a parallel Toffoli layer, uses at least $6m+1$ $T$ gates. Campbell and Howard gave
a circuit with $6m+1$, and to our knowledge this is the first proof that it is optimal for
general $m$. The proof combines Campbell and Howard's self-orthogonality of the columns of the parity matrix
with a rank (flattening) bound, sharpened by one. More generally, the argument gives a lower bound on the $T$-count of every
diagonal gate in the third level of the Clifford
hierarchy, the isotropy floor, computable by Gaussian elimination. It is never below the stabilizer
nullity $\nu$ and reaches $2\nu+1$ on non-Clifford gates whose phase is purely cubic. No bound that also holds for circuits with measurement and feedforward can do this,
since such circuits implement $\CCZ$ ($\nu=3$) with four $T$ gates and the layer with $4m$,
where the floor gives seven and $6m+1$. The floor also recovers Campbell and Howard's count $4m+3$ for a
fan-out of Toffolis. On the output of PyZX's TODD-based optimizer, relative to its Hadamard placement, it
certifies $193$ of $311$ phase-polynomial blocks $T$-optimal once repeated parities are merged (at the
best of five optimizer seeds), against $113$ for nullity, including
$37$ of the $72$ blocks too large for exhaustive search.
\end{abstract}

\setcounter{tocdepth}{1}
\tableofcontents
\bigskip

\input{sections/01-intro}         %
\input{sections/02-overview}      %
\input{sections/03-related}       %
\input{sections/04-discussion}    %

\appendix

\input{appendices/A-prelim}        %
\input{appendices/B-floor}         %
\input{appendices/C-exact}         %
\input{appendices/D-applications}  %

\section*{Acknowledgments}

The author thanks Michele Mosca and Vlad Gheorghiu for hosting and mentoring this work through an
undergraduate research award at the University of Waterloo in the summer of 2026, and Matthew Amy for
a helpful discussion that motivated the topic.

\input{sections/99-disclosure}    %

\bibliographystyle{unsrt}
\bibliography{main}

\end{document}

%% file: sections/01-intro.tex
\section{Introduction}\label{sec:intro}

Each $T$ gate in a fault-tolerant Clifford$+T$ computation consumes a distilled magic state, and in
surface-code layouts the distillation factories take a large share of the qubits and the
runtime~\cite{Litinski}. A compiled circuit gives an upper bound on the $T$-count, but only for the
compiler that produced it, and to know how much room any compiler has left one needs a lower bound
that is cheap to compute and close to the true count.

Toffoli gates account for much of the $T$-count of reversible arithmetic, phase oracles and the
control logic of select operations, and they rarely come alone. For a layer of $m$ Toffolis on
disjoint qubits, Campbell and Howard~\cite{CampbellHowardPRA} found a circuit with $6m+1$ $T$
gates, against $7m$ for compiling the Toffolis one at a time. At $m=1$ the value $7$ is classical,
and Campbell and Howard report the value $13$ at $m=2$ as known, a six-qubit case that exhaustive
decoding settles. For $m\ge3$ optimality was open, and the best lower bounds we know of grow
like $3m$ (stabilizer nullity~\cite{JiangWang}) or about $3.71m$ (the amortized stabilizer R\'enyi
entropy bound of Zhu et al.~\cite{AmortizedSRE}, which we evaluate exactly in Table~\ref{tab:rate}). This paper proves that $6m+1$ is optimal for every $m$ among Hadamard-free circuits.

\begin{theorem}[the parallel Toffoli layer \informal{thm:parccz}]\label{inf-thm:parccz}
A layer of $m$ disjoint $\CCZ$ gates, the diagonal core of a layer of $m$ parallel Toffolis, has
$T$-count exactly $6m+1$ over $\{\CNOT,T\}$ circuits.
\end{theorem}

The model here is the Hadamard-free one: circuits of \CNOT{} and $T$ gates, with diagonal Clifford
gates free. It is the natural model for diagonal gates, the one in which Amy and Mosca's
characterization below holds, and the one in which phase-polynomial optimizers such as TODD~\cite{HeyfronCampbell}
work, as does AlphaTensor-Quantum~\cite{AlphaTensorQuantum} when its ancilla gadgets are switched off. Clean ancillas do
not lower the $T$-count in this model (Corollary~\ref{cor:ancilla}), so every bound below also
covers circuits that use them. It differs from the ancilla-free Clifford$+T$
model, which allows internal Hadamards but no ancillas, and in which Gosset et al.\ showed by
exhaustive search that a single Toffoli needs seven $T$ gates~\cite{GKMR}. Measurement changes the
count. With temporary AND gates and measurement~\cite{Jones,Gidney2018} the layer costs about $4m$
$T$ gates, so the bound does not limit fault-tolerant compilers that use them. It limits the
resynthesis of Hadamard-free phase blocks, which is the step TODD and gadget-free AlphaTensor-Quantum perform
and the step the certificates of \S\ref{sec:benchintro} audit.
The supplementary notes~\cite{IsofloorSupp} discuss how far the bounds extend to circuits with
internal Hadamards, approximation and measurement.

\subsection{Existing lower bounds}

For diagonal gates the exact $T$-count is known in principle. Amy and Mosca~\cite{AmyMosca} showed
that it equals the distance from a fixed word to a punctured Reed--Muller code, and no
polynomial-time algorithm is known for that decoding problem. In coding theory the cosets of
Reed--Muller codes have been studied mainly through covering radii and through classifications of
cubic forms modulo quadratics~\cite{BrierLangevin,Khoruzhii2026}, which parametrize the cubic part
of the cosets of $\RM(n-4,n)^{\ast}$. To our knowledge, exact minimum weights for direct sums of
cubic monomials had not been determined. For general Clifford$+T$ circuits
optimizing the $T$-count is NP-hard~\cite{vandeWeteringAmy}. Meet-in-the-middle search~\cite{AMMR}
and exact counting~\cite{GheorghiuMoscaMukhopadhyay} work for a few qubits only.

Magic monotones give lower bounds without a search. Robustness of magic~\cite{HowardCampbellRoM},
channel robustness~\cite{SeddonCampbell}, stabilizer extent~\cite{StabilizerExtent} and stabilizer
R\'enyi entropies~\cite{LeoneSRE,AmortizedSRE} all bound the $T$-count, but computing them on a
general gate takes exponential time. Stabilizer nullity~\cite{BCHK} is the exception on the gates
studied here (Proposition~\ref{prop:nullitygen} computes it by Gaussian elimination), but it gives
$3$ for $\CCZ$, whose $T$-count is $7$. A quantity that cannot increase under stabilizer operations
with feedforward also bounds circuits that measure and correct. Such circuits implement $\CCZ$ with
four $T$ gates~\cite{Jones}, so an argument that reaches $7$ has to use unitarity
(Remark~\ref{rem:nomonotone}).

\subsection{The isotropy argument}

A $\{\CNOT,T\}$ circuit applies each
of its $t$ $T$ gates to a parity of the input bits. Record, for each qubit $a$, the vector
$c_a\in\F_2^t$ that marks which $T$ gates involve $a$. When the target's phase has no linear or
quadratic part, a condition we call pure-cubic, every $c_a$ has even weight and every pair has even
overlap. This is the case of vanishing quadratic part of Campbell and Howard's
Lemma~3~\cite{CampbellHowardPRA}, and it says
that the span of the $c_a$ is a self-orthogonal subspace of $\F_2^t$. A self-orthogonal subspace
has dimension at most $t/2$, so $t$ is at least twice the rank of the parity matrix.

The rank, in turn, is bounded below by the cubic part of the phase. The rank of any flattening of
the signature tensor of Heyfron and Campbell~\cite{HeyfronCampbell} bounds its symmetric rank from
below, a standard fact, and in our language the rank of the parity matrix is at least $n-d$, where $d$ is the
dimension of the radical of the cubic moment form, the directions along which the cubic part of
the phase degenerates. Combining the two bounds gives $t\ge2(n-d)$, twice what nullity gives, and
a parity argument adds one more when the gate is not Clifford.
Campbell and Howard use self-orthogonality to compute a cost that sets $\CCZ$ factors aside and is
zero on every product of $\CCZ$ gates~\cite[Def.~2, Lemma~3]{CampbellHowardPRA}. The flattening
bound applied to $t$ itself gives only $t\ge n-d$, the nullity of a pure-cubic gate. The doubling
appears only once the flattening bound is applied to the rank of the parity matrix instead. We call the result the isotropy
floor. For $m$ disjoint $\CCZ$ gates $d=0$ and $n=3m$, so the floor is $6m+1$. Campbell and Howard
do determine the $T$-count of every controlled diagonal Clifford exactly~\cite[Thm.~5]{CampbellHowardPRA},
but for $m\ge2$ the layer is not a controlled unitary in any \CNOT{} basis
(Remark~\ref{rem:layernotcontrolled}), so their theorem does not reach it. Beyond the pure-cubic case,
the general floor and the identity $\nu=n-d^{\ast}$ below carry the rest of the paper, including every benchmark certificate, since none of the certified
blocks is pure-cubic.

The weights, overlaps and triple counts of the columns are the moments of the gate, and they do not
depend on the circuit. Every quantity in the floor is a rank or a kernel dimension of a matrix built
from them, so Gaussian elimination computes it without a search over circuits. For a gate given as
a $\{\CNOT,T\}$ circuit the moments follow from
propagating parities through the \CNOT{}s. For a gate given as a Boolean function (a Toffoli
network, a phase oracle or a QAOA cost layer) they follow from the function directly.

\subsection{Contributions}

The diagonal gates of the third Clifford level are those whose phase polynomial over $\Z_8$ has
degree at most three. A $T$ gate contributes a linear term, $\CS$ a quadratic one and $\CCZ$ a cubic
one. We write $\delta(U)$ for the least number of $T$ gates in any $\{\CNOT,T\}$ circuit for $U$.

\begin{enumerate}
\item Every diagonal level-three gate on $n$ qubits has
$\delta(U)\ge2(n-d^{\ast})-r$ (Theorem~\ref{thm:genLB}). Here $d^{\ast}$ is the radical dimension
of the full moment tensor, the cubic moment form together with the lower moments, and $r$ is the
rank of the quadratic moments. Both take $O(n^3)$ word operations. The floor is never below the
stabilizer nullity $\nu$, which equals $n-d^{\ast}$ on this class, and on non-Clifford pure-cubic
gates it is $2\nu+1$ (Corollaries~\ref{cor:floordominates} and~\ref{cor:cubicplusone}). An
upper bound computed in the same time brackets every gate (Proposition~\ref{prop:ceiling}).

\item Besides the parallel layer, the floor is exact on a layer whose blocks share single qubits,
when the shared qubits can be routed along a forest in which each block has at most three neighbors
(Proposition~\ref{prop:splicetree}). There the count is $6m+1-2\eta$, where $\eta$ counts the excess
incidences. For $m$ Toffolis fanned out from one control the floor recovers Campbell and Howard's
exact count $4m+3$ (Theorem~\ref{thm:fanout}).  On a general controlled
diagonal Clifford the floor is within one of their count (Proposition~\ref{prop:controlled}).

\item PyZX~\cite{PyZX} implements a phase-block optimizer based on Heyfron and Campbell's
tensor-decomposition algorithm TODD~\cite{HeyfronCampbell}. On its output for $24$ benchmark circuits
the floor proves $193$ of $311$ blocks optimal at the best of five seeds ($186$ to $193$ across
seeds), relative to the optimizer's block partition, against
$113$ for nullity, and $37$ of the $72$ blocks too large to decode exactly, against $15$. On all
$515$ blocks small enough to decode, across two optimizers, it is within two of the true $T$-count
(\S\ref{sec:benchintro}).
\end{enumerate}

The floor also gives $6k+1$ for the $\mathrm{GF}(2^k)$ multiplication oracle, against $3k$ for
nullity (\S\ref{sec:oracleintro}). Further results are in the supplementary notes~\cite{IsofloorSupp}: a census of pure-cubic gates
with $n\le6$, exact counts for cubics that factor, permutation gates and further applications,
higher levels of the hierarchy and circuits beyond the Hadamard-free model.

In Section~\ref{sec:overview} we go through the whole argument for a single $\CCZ$ before stating
the general results. Related work is discussed in Section~\ref{sec:related} and open problems in
Section~\ref{sec:outlook}. All proofs are in the appendices.

%% file: sections/02-overview.tex
\section{Technical overview}\label{sec:overview}

We start with a single $\CCZ$. The argument that its $T$-count is seven already contains every idea of the general bound, so we give it in full before stating the general
results.

\subsection{Warm-up: the $T$-count of $\CCZ$}\label{sec:warmup}\label{sec:mechanism}

The gate $\CCZ$ maps $\ket{x}\mapsto(-1)^{x_1x_2x_3}\ket{x}$. A $\{\CNOT,T\}$ circuit can only
produce phases that are eighth roots of unity, because $T=\mathrm{diag}(1,e^{2\pi i/8})$, and
$(-1)^b=e^{2\pi i(4b)/8}$ is one of them. We therefore describe a diagonal gate by its phase polynomial
$\phi\colon\F_2^n\to\Z_8$, where $\ket{x}$ picks up $e^{2\pi i\phi(x)/8}$. For $\CCZ$,
\begin{equation}\label{eq:cczphi}
  \phi(x)\;=\;4\,x_1x_2x_3 .
\end{equation}

A $T$ gate applied to the parity $y\cdot x$ of some input bits multiplies the state by
$e^{2\pi i(y\cdot x)/8}$, reading the bit $y\cdot x$ as the integer $0$ or $1$. A circuit for $\phi$
amounts to a choice of coefficients $\alpha_y\in\Z_8$, one per parity, with
$\sum_y\alpha_y\,(y\cdot x)\equiv\phi(x)$ for all $x$. For $\CCZ$ one choice is the inclusion-exclusion
identity
\begin{equation}\label{eq:cczid}
  4x_1x_2x_3\;\equiv\;x_1+x_2+x_3-(x_1\oplus x_2)-(x_1\oplus x_3)-(x_2\oplus x_3)
  +(x_1\oplus x_2\oplus x_3)\pmod 8 .
\end{equation}
Since $S$ adds $2$ to a coefficient and $Z$ adds $4$, only the odd coefficients need a $T$ gate, and
the coefficient $-1\equiv7$ is a $T^{\dagger}$. The seven terms give the familiar seven-$T$ circuit,
in which \CNOT{}s move each parity onto a wire and back:
\begin{center}
\begin{tikzpicture}[x=8mm,y=-9mm,
  gate/.style={draw,fill=white,rounded corners=1pt,minimum width=5.4mm,minimum height=4.6mm,
               inner sep=1pt,font=\scriptsize},
  dot/.style={circle,fill,inner sep=1.3pt},
  targ/.style={circle,draw,minimum size=3.1mm,inner sep=0pt,path picture={
      \draw (path picture bounding box.west)--(path picture bounding box.east);
      \draw (path picture bounding box.north)--(path picture bounding box.south);}}]
  \foreach \r/\lab in {0/{q_1},1/{q_2},2/{q_3}} {
    \draw (0.35,\r) -- (15.85,\r);
    \node[font=\scriptsize,left] at (0.35,\r) {$\lab$};
  }
  \foreach \r in {0,1,2} { \node[gate] at (1,\r) {$T$}; }
  \foreach \x/\a/\b in {2/0/1, 4/0/1, 5/1/2, 7/1/2, 8/0/2, 10/0/2,
                        11/0/1, 12/1/2, 14/1/2, 15/0/1} {
    \draw (\x,\a) -- (\x,\b);
    \node[dot] at (\x,\a) {};
    \node[targ] at (\x,\b) {};
  }
  \node[gate] at (3,1)  {$T^{\dagger}$};
  \node[gate] at (6,2)  {$T^{\dagger}$};
  \node[gate] at (9,2)  {$T^{\dagger}$};
  \node[gate] at (13,2) {$T$};
  \node[font=\tiny] at (1,0.52)  {$x_1$};
  \node[font=\tiny] at (1,1.52)  {$x_2$};
  \node[font=\tiny] at (1,2.52)  {$x_3$};
  \node[font=\tiny] at (3,1.52)  {$x_1\oplus x_2$};
  \node[font=\tiny] at (6,2.52)  {$x_2\oplus x_3$};
  \node[font=\tiny] at (9,2.52)  {$x_1\oplus x_3$};
  \node[font=\tiny] at (13,2.52) {$x_1\oplus x_2\oplus x_3$};
\end{tikzpicture}
\end{center}
The parity each $T$ acts on is written beneath it.

Other choices of coefficients give other circuits. The sets of parities whose $T$ gates multiply
to a diagonal Clifford form a linear code, a punctured Reed--Muller code~\cite{AmyMosca}, so the
circuits for a gate form one coset of that code and $\delta(U)$ is the minimum weight in the coset. Finding that weight is a decoding problem with no known polynomial-time
algorithm, and the rest of this section bounds it from below without solving it.

\subsubsection{Isotropy}

Fix any circuit and list its $t$ odd parities as the rows of a matrix $Y$ over $\F_2$, one row per
$T$ gate and one column per qubit. For the circuit above, taking the rows in the order of
\eqref{eq:cczid},
\begin{equation}\label{eq:cczY}
Y^{\mathsf T} \;=\;
\begin{pmatrix}
1&0&0&1&1&0&1\\
0&1&0&1&0&1&1\\
0&0&1&0&1&1&1
\end{pmatrix}.
\end{equation}
Write $c_a\in\F_2^t$ for column $a$ of $Y$, so that $c_a$ marks the $T$ gates that involve qubit
$a$. The weight of $c_a$ counts the parities containing $a$, and the overlap
$\langle c_a,c_b\rangle$ counts those containing both $a$ and $b$. Their values mod $2$ are the
order-one and order-two moments of the gate, and since by Fact~\ref{fact:moment} they are the same
for every circuit implementing it, any one circuit computes them. We call a phase
\emph{pure-cubic} when all of its order-one and order-two moments vanish. $\CCZ$ is pure-cubic, since
every column has weight four and every pair of columns overlaps in two places. For a
pure-cubic gate, then,
\begin{equation}\label{eq:iso}
  \langle c_a,c_b\rangle\;=\;0\qquad\text{for all }a,b,\ \text{including }a=b.
\end{equation}
The column space $\mathcal C=\mathrm{span}\{c_1,\dots,c_n\}$ is therefore self-orthogonal,
$\mathcal C\subseteq\mathcal C^{\perp}$. Since $\dim\mathcal C^{\perp}=t-\dim\mathcal C$,
\begin{equation}\label{eq:half}
  t\;\ge\;2\dim \mathcal C\;=\;2\rank Y .
\end{equation}

\subsubsection{The rank, and the extra one}

To bound $\rank Y$ we count parities three qubits at a time. Let $T_{abc}$ be the parity of the
number of rows of $Y$ that contain all of $a$, $b$ and $c$. For distinct $a,b,c$ this is
\begin{equation}\label{eq:Tabc}
  T_{abc}\;=\;\sum_{j=1}^{t}(c_a)_j\,(c_b)_j\,(c_c)_j\pmod 2 ,
\end{equation}
and, like the lower moments, it depends only on the gate. These numbers form a symmetric trilinear
form, the cubic moment form. For $\CCZ$ only one row, $x_1\oplus x_2\oplus x_3$, contains all three
qubits, so $T_{123}=1$ is the only nonzero entry. Fixing one argument in a direction
$w\in\F_2^n$ gives a matrix $T(w)_{bc}=\sum_a w_aT_{abc}$, and the radical $R$ is the set of $w$
with $T(w)=0$. For $\CCZ$,
\begin{equation*}
  T(w)\;=\;\begin{pmatrix}0&w_3&w_2\\ w_3&0&w_1\\ w_2&w_1&0\end{pmatrix},
\end{equation*}
which vanishes only at $w=0$, so $R=\{0\}$ and $d=\dim R=0$.

Now suppose a combination of columns vanishes, $\sum_a w_ac_a=0$. Substituting into
\eqref{eq:Tabc} gives $T(w)_{bc}=\sum_j(c_b)_j(c_c)_j\bigl(\sum_a w_ac_a\bigr)_j=0$, so $w\in R$.
The kernel of $Y$ is contained in the radical, and
\begin{equation}\label{eq:rank}
  \rank Y\;\ge\;n-d .
\end{equation}
Together with \eqref{eq:half} this gives $t\ge2(n-d)$ for every circuit, which is $6$ for $\CCZ$.

Could a circuit for $\CCZ$ use exactly six $T$ gates? Its three columns would be independent by
\eqref{eq:rank}, so $\mathcal C$ would have dimension three, half of six, and \eqref{eq:half} would
be tight: $\mathcal C=\mathcal C^{\perp}$. Every column has even weight, so the all-ones vector
$\mathbf 1\in\F_2^6$ is orthogonal to every column and lies in $\mathcal C^{\perp}=\mathcal C$. Then
$\mathbf 1=\sum_{a\in S}c_a$ for some nonempty set $S$ of qubits. Substituting this into
\eqref{eq:Tabc} with $w$ the indicator of $S$ gives $T(w)_{bc}=\langle c_b,c_c\rangle=0$, so $w$ is a
nonzero element of the radical. That contradicts $R=\{0\}$ (Lemma~\ref{lem:allones}). Hence
$\delta(\CCZ)\ge7$, and the circuit above shows equality. On three qubits this value is classical,
but the argument never decodes, so it runs unchanged on any number of qubits.

The argument used purity twice, for the self-orthogonality \eqref{eq:iso} and to remove the
coincident-index terms in \eqref{eq:Tabc}. For a general phase the full moment tensor and its
radical $R^{\ast}$ take the place of $T$ and $R$. The extra one used $R=\{0\}$. When $d>0$ we first
split off the radical directions, and the argument applies to what remains, which recovers the
extra one for every non-Clifford pure-cubic gate (Corollary~\ref{cor:cubicplusone}). Figure~\ref{fig:pipeline} summarizes the computation.

\input{figures/fig1-pipeline}

\subsubsection{Many gates at once}\label{sec:layerintro}

The same computation applies to $m$ disjoint $\CCZ$s on $n=3m$ qubits. Their cubic moment form is
block-diagonal, each block has a trivial radical, and so $d=0$ and the floor is $2n+1=6m+1$. The
construction of Campbell and Howard~\cite[Cor.~1]{CampbellHowardPRA} meets it
(Proposition~\ref{prop:splicetree}, drawn at $m=2$ in Figure~\ref{fig:layer}), which proves
Theorem~\ref{inf-thm:parccz}.

Stabilizer nullity adds over tensor products and gives only $3m$ for the layer.
Table~\ref{tab:rate} compares the three bounds.

\begin{table}[t]
\centering
\caption{Lower bounds on a layer of $m$ disjoint $\CCZ$ gates. The amortized stabilizer R\'enyi
entropy (SRE) bound of Zhu et al.~\cite{AmortizedSRE} has rate
$\sigma=(5-\log_211)/(2-\log_23)=3.7118\dots$ per $\CCZ$, the amortized entropy of $\CCZ$ divided by
that of one $T$ gate. This value of the bound is computed exactly, not estimated (\S\ref{sec:related}). The isotropy floor meets the construction for every
$m$.}
\label{tab:rate}
\begin{tabular}{@{}lccccc@{}}
\toprule
target & $n$ & $\nu$~\cite{JiangWang} & amortized SRE~\cite{AmortizedSRE} & isotropy floor
& $\delta$\\
\midrule
$\CCZ$                    & $3$  & $3$  & $4$                       & $7$    & $7$\\
two disjoint $\CCZ$       & $6$  & $6$  & $8$                       & $13$   & $13$\\
three disjoint $\CCZ$     & $9$  & $9$  & $12$                      & $19$   & $19$\\
$m$ disjoint $\CCZ$       & $3m$ & $3m$ & $\lceil\sigma m\rceil$ & $6m+1$ & $6m+1$\\
\midrule
rate per $\CCZ$           &      & $3.00$ & $3.712$                 & $6.00$ & $6.00$\\
\bottomrule
\end{tabular}
\end{table}

\subsection{The general floor}\label{sec:floorladder}\label{sec:landscape}

Many gates that arise in practice, such as a Toffoli network compiled together with nearby $S$ and
$T$ gates, have order-one and order-two moments that do not vanish. We keep those lower moments in a
matrix.

\begin{definition}[cubic moment form and radical]\label{def:cmf}
For distinct $a$, $b$, $c$ the order-three moment $T_{abc}$ is the parity of the number of odd
parities containing all three, and $T_{abc}=0$ when two indices coincide. The resulting symmetric
trilinear form on $\F_2^n$ is the \emph{cubic moment form}. Its \emph{slice} in direction
$w\in\F_2^n$ is the matrix $T(w)_{bc}=\sum_a w_aT_{abc}$, and its \emph{radical} is
$R=\{w:T(w)=0\}$, with $d=\dim R$. The moments of orders one and two form the symmetric
\emph{quadratic moment matrix} $Q$, of rank $r$. Placing them at coincident indices gives the full
moment tensor, whose radical $R^{\ast}$ has dimension $d^{\ast}$.
\end{definition}

The full moment tensor is the signature tensor of Heyfron and Campbell~\cite{HeyfronCampbell}. In
their language $\delta$ is its symmetric rank over
$\F_2$, and what follows is a lower bound on that rank.

\begin{theorem}[the isotropy floor \informal{thm:genLB}]\label{inf-thm:floor}
Every diagonal level-three gate $U$ on $n$ qubits satisfies $\delta(U)\ge2(n-d^{\ast})-r\ge\nu$,
where $\nu=n-d^{\ast}$ is its stabilizer nullity. Both $d^{\ast}$ and $r$ are computed from the
moments in $O(n^3)$ word operations.
\end{theorem}

On a pure-cubic gate $r=0$ and the theorem gives $2\nu$, to which the parity step of the warm-up
adds one when the gate is not Clifford. We use \emph{isotropy floor} for $2\nu+1$ in that case and
for $2(n-d^{\ast})-r$ otherwise. The quadratic moments weaken the bound. The floor exceeds nullity by
$\nu-r$ (Corollary~\ref{cor:floordominates}), so a gate with $r=\nu$ gets nothing beyond nullity.

\subsubsection{Where the floor is exact}\label{sec:sharing}

Suppose any two blocks share at most one qubit, and let $\eta=3m-n$ count the excess incidences.
Such a layer is still non-degenerate, so its floor is $2n+1=6m+1-2\eta$
(Lemma~\ref{lem:linearoverlap}). On a shared wire two $T$ gates of opposite sign cancel
(Figure~\ref{fig:merge}), and by splicing circuits together along the shared qubits we can meet the
floor whenever, for each shared qubit, the blocks containing it can be joined by a tree so that the
union of these trees is a forest in which every block has at most three neighbors
(Proposition~\ref{prop:splicetree}). A fan-out of $m$ Toffolis from one control is of this kind, and
there the floor recovers Campbell and Howard's exact count $4m+3$~\cite[Thm.~5]{CampbellHowardPRA}.
 Once the sharing closes a cycle the floor can be loose.
The smallest example is three blocks that meet pairwise in different qubits,
$\{0,1,2\}$, $\{0,3,4\}$ and $\{1,3,5\}$, with floor $13$ and $T$-count $15$
(Remark~\ref{rem:ledgerfails}).

Solving the moment equations from the top down, triples first, produces an explicit circuit, so
every gate also gets an upper bound $W(s)$ in $O(n^3)$ time (Proposition~\ref{prop:ceiling}).

\subsection{Applications}\label{sec:oracleintro}\label{sec:benchintro}

The multiplication oracle of $\mathrm{GF}(2^k)$ acts as
$U_k\ket{x,y,z}=(-1)^{z\cdot(x\ast y)}\ket{x,y,z}$ on $3k$ qubits. It is non-degenerate, so
$\delta(U_k)\ge6k+1$ (Proposition~\ref{prop:fieldoracle}). Since bilinear complexity gives circuits with $O(k)$
$T$ gates, the count is $\Theta(k)$. At $k=2$ exact decoding gives $17$, the count of the best
published circuits~\cite{Vandaele}. AlphaTensor-Quantum had already shown with an SMT solver that no
decomposition of the signature tensor is shorter~\cite[\S3.1]{AlphaTensorQuantum} (Remark~\ref{rem:gf4}).

The floor can also be applied to each phase block of an optimizer's output. We call an output block
\emph{certified} when its floor equals its $T$-count after repeated parities in the block are
merged, since then no
$\{\CNOT,T\}$ resynthesis of that block can do better. On the output of PyZX's TODD-based
phase-block optimizer~\cite{HeyfronCampbell,PyZX} for the $24$ standard reversible
benchmarks~\cite{AMMR,AmyMaslovMoscaTpar} that it finishes within four minutes, the floor certifies
$193$ of $311$ blocks and nullity $113$ at the best of five optimizer seeds ($186$ to $193$ across
seeds). Of the $72$ blocks with $7$ to $28$ active qubits, too
large for exact decoding, the floor certifies $37$ and nullity $15$. 
The certificates are relative to the optimizer's choice of Hadamard positions and blocks
(Appendix~\ref{app:bench}).

%% file: figures/fig1-pipeline.tex
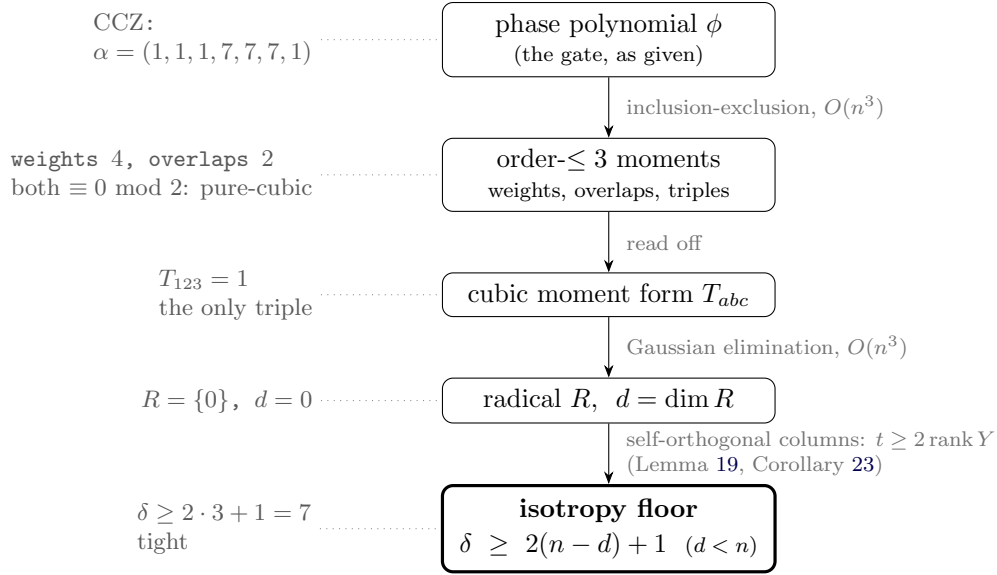
\begin{figure}[t]
\centering
\begin{tikzpicture}[
  >={Stealth[scale=1.0]},
  node distance=8mm,
  box/.style={draw,rounded corners,align=center,inner sep=4pt,font=\small,minimum width=44mm},
  inst/.style={align=left,inner sep=3pt,font=\footnotesize\ttfamily,text=black!65},
  cost/.style={font=\scriptsize,text=black!55,midway,right=1mm,align=left},
  floorbox/.style={draw,rounded corners,align=center,inner sep=5pt,font=\small,
              minimum width=44mm,very thick}
]

\node[box] (phi) {phase polynomial $\phi$\\[-1pt]{\scriptsize(the gate, as given)}};
\node[box,below=of phi] (mom) {order-$\le3$ moments\\[-1pt]{\scriptsize weights, overlaps, triples}};
\node[box,below=of mom] (ten) {cubic moment form $T_{abc}$};
\node[box,below=of ten] (rad) {radical $R$,\ \ $d=\dim R$};
\node[floorbox,below=of rad]  (flr) {\textbf{isotropy floor}\\[1pt]
      $\delta\ \ge\ 2(n-d)+1$\ \ {\scriptsize($d<n$)}};

\draw[->] (phi) -- (mom) node[cost] {inclusion-exclusion, $O(n^3)$};
\draw[->] (mom) -- (ten) node[cost] {read off};
\draw[->] (ten) -- (rad) node[cost] {Gaussian elimination, $O(n^3)$};
\draw[->] (rad) -- (flr) node[cost] {self-orthogonal columns: $t\ge2\rank Y$\\(Lemma~\ref{lem:isocol}, Corollary~\ref{cor:cubicplusone})};

\node[inst,left=16mm of phi] (i1) {$\CCZ$:\\ $\alpha=(1,1,1,7,7,7,1)$};
\node[inst,left=16mm of mom] (i2) {weights $4$, overlaps $2$\\ \textrm{both $\equiv0$ mod $2$: pure-cubic}};
\node[inst,left=16mm of ten] (i3) {$T_{123}=1$\\ \textrm{the only triple}};
\node[inst,left=16mm of rad] (i4) {$R=\{0\}$,\ $d=0$};
\node[inst,left=16mm of flr] (i5) {$\delta\ge 2\cdot3+1=7$\\ \textrm{tight}};

\foreach \a/\b in {i1/phi,i2/mom,i3/ten,i4/rad,i5/flr}
  \draw[dotted,black!45] (\a) -- (\b);

\end{tikzpicture}
\caption{The isotropy floor, end to end, with the \S\ref{sec:warmup} example beside it. The
figure shows the pure-cubic case, and \S\ref{sec:landscape} gives the general floor. At $\CCZ$ the floor returns
$7$ and the seven-$T$ construction meets it, so $\delta(\CCZ)=7$.}
\label{fig:pipeline}
\end{figure}

%% file: sections/03-related.tex
\section{Related work}\label{sec:related}

Beverland, Campbell, Howard and Kliuchnikov~\cite{BCHK} introduced stabilizer nullity as a state
monotone, which also bounds circuits with adaptive feedforward, and Jiang and Wang~\cite{JiangWang}
defined its unitary version. Apart from the related monotones of~\cite{BCHK} and the flattening
bound on the symmetric rank discussed below, we know of no other lower bound that is efficiently
computable on all diagonal level-three gates. Jiang and Wang also
show that the nullity of a unitary can be twice that of a state it prepares. The Choi state has no
such gap, since its nullity equals that of the unitary~\cite{JiangWang}. That factor of two is unrelated to ours, which compares the
$T$-count with the unitary nullity. We prove $\delta\ge2\nu+1$ on every non-Clifford pure-cubic gate
(Corollary~\ref{cor:cubicplusone}), with equality on layers whose blocks share single qubits along a forest. The supplementary
notes~\cite{IsofloorSupp} add equality on cubics that factor and on every pure-cubic gate with
$n\le5$.

Campbell and Howard~\cite{CampbellHowardPRA} gave the $6m+1$ construction for $m$ disjoint
$\CCZ$ gates~\cite[Thm.~6, Cor.~1]{CampbellHowardPRA}. Their controlled-unitary
theorem~\cite[Thm.~5]{CampbellHowardPRA} gives, through Lempel's factorization, the exact $T$-count
of every controlled diagonal Clifford among ancilla-free circuits. This class includes layers whose
blocks all share one control, and on it our floor comes within one of their count
(Proposition~\ref{prop:controlled}). It excludes the disjoint layer for $m\ge2$ and every other
layer of blocks sharing at most one qubit (Remark~\ref{rem:layernotcontrolled}). The step that
needs the control is in their converse~\cite[App.~C]{CampbellHowardPRA}: it splits the parities by
whether they touch the control qubit and bounds each half by an optimal Lempel factorization of the
target's quadratic form, which requires a qubit through which every cubic monomial passes. Their Lemma~3 is the self-orthogonality we start from.
For disjoint layers, the earlier lower bounds were the nullity $3m$~\cite{JiangWang} and the
amortized stabilizer R\'enyi entropy bound of Zhu et al.~\cite{AmortizedSRE}, which gives about
$3.71m$ (Table~\ref{tab:rate}). This is the exact value of that bound on the layer, not an estimate. The amortized entropy of a single $\CCZ$ equals
the entropy $5-\log_211$ of its Choi state~\cite[App.~D, Thm.~S13]{AmortizedSRE}. The amortized entropy is
subadditive under composition, because ancillas are included in its supremum, so on the layer it
is at most $m(5-\log_211)$, and it is at least the Choi-state entropy~\cite[Cor.~4]{AmortizedSRE},
which is additive over the blocks. The bound is proved for unitary circuits. For general gates we evaluate only the
Choi-state form of the bound, the strongest exponential-time measure we compared with
(see the supplementary notes~\cite{IsofloorSupp}). Seddon and Campbell~\cite{SeddonCampbell} define channel
robustness and magic capacity, which also bound the $T$-count of multi-qubit operations. Both are
optimizations over stabilizer decompositions whose size grows superexponentially with $n$.

Heyfron and Campbell~\cite{HeyfronCampbell} phrase $T$-count optimization as decomposing the
signature tensor of a phase polynomial into few symmetric rank-one terms. Their TODD compiler and
AlphaTensor-Quantum~\cite{AlphaTensorQuantum} search for such decompositions, and Khoruzhii, Gel\ss{}
and Pokutta~\cite{KhoruzhiiTD2026} match or improve the counts of these searches with tensor
decomposition over $\F_2$ at much lower cost. The signature tensor is our full moment tensor, so
$\delta$ is its symmetric rank and the floor is a lower bound on that rank. The rank of a
flattening also bounds the symmetric rank from below, which is our inequality $\rank Y\ge n-d$.
We combine it with isotropy, which doubles the bound. For symmetric
matrices the minimal decomposition over $\F_2$ is known from Lempel~\cite{Lempel1975} (see also
Seroussi and Lempel~\cite{SeroussiLempel}), and the $-r$
term of Theorem~\ref{thm:genLB} uses the same count of alternating forms. AlphaTensor-Quantum found
$17$ for the $\mathrm{GF}(4)$ multiplier and showed with an SMT solver that no decomposition of its
signature tensor is shorter~\cite[\S3.1]{AlphaTensorQuantum}. Khoruzhii, Gel\ss{} and
Pokutta~\cite{Khoruzhii2026} classify cubic forms by their alternating rank, an invariant different
from the radical used here.

In Amy and Mosca's characterization~\cite{AmyMosca} the $T$-count is a decoding problem, and the
floor bounds its optimum from below without solving it. Kam et al.~\cite{KamExactness2026} bound the
$T$-count left after exact post-processing by the denominator exponent of the matrix entries. That
certificate gives nothing for the gates studied here, whose entries are eighth roots of unity and
have no denominator.

Girish, May, Parham and Yuen~\cite{MagicComm} use communication complexity to prove that the
$n$-controlled Toffoli needs $\Omega(n)$ magic, and Gosset, Kothari and Wu~\cite{GKW24} bound state
preparation through Pauli postselection. These bounds are asymptotic, and the first also covers
adaptive circuits if measurements are counted together with magic. Gosset, Kothari and
Zhang~\cite{GKZToffoli} bound the $T$-count of a reversible Boolean function below by its
non-adaptive parity decision tree complexity minus one, for unitary circuits and, at zero error,
for adaptive ones. That complexity is at most the number $n$ of inputs, so the bound stays below the
nullity $n$ of a non-degenerate cubic gate. Our bounds are per gate and
exact on specific families, but the factor of two over nullity needs the absence of feedforward
(see the supplementary notes~\cite{IsofloorSupp}).

In magic-state distillation and code design, Bravyi and Haah~\cite{BravyiHaah}, Nezami and
Haah~\cite{NezamiHaah} and Rengaswamy et al.~\cite{RengaswamyCSST} use related self-orthogonality
conditions (triorthogonality and the CSS-T condition) to control distillation rates and
transversal $T$ gates. Those works use the conditions to constrain code parameters. The closest in method is Nezami and
Haah's bound $n\ge2k$ for triorthogonal codes~\cite[Lemma~7]{NezamiHaah}, which counts orthonormal
vectors supported on half of the coordinates, where the isotropy step bounds the dimension of a
self-orthogonal span. We use only the
weaker pairwise version (even column weights and even pairwise overlaps) and obtain from it a lower
bound on the $T$-count of a given gate. Read on the columns of $Y$, triorthogonality asks every
triple overlap to be even, that is $T=0$, which makes a pure-cubic gate Clifford. Generalized
triorthogonality~\cite{HaahHastings} and synthillation~\cite{CampbellHowardPRA} allow the odd triple
overlaps that encode $\CCZ$, but they are used to construct codes and circuits, and we know of no
lower bound from them that reaches the layer.

%% file: sections/04-discussion.tex
\section{Discussion and open problems}\label{sec:outlook}

The isotropy floor is not monotone under stabilizer operations with feedforward, which is why it
can exceed the factor $4/3$ over nullity that limits every such monotone on $\CCZ$. The cost is that
the bound depends on the circuit model. We prove it here for Hadamard-free circuits. The
supplementary notes~\cite{IsofloorSupp} extend it to some circuits with internal Hadamards and to
higher levels of the Clifford hierarchy, show that under adaptive feedforward only nullity is known
to hold, and collect the open problems on those models.

Within the Hadamard-free model the floor is exact on the layers of Appendix~\ref{app:exact}, and
Conjecture~\ref{conj:ledger} states when the sharing floor is attained. Open
Problems~\ref{op:hadamard}--\ref{op:controlled} collect the main gaps that remain.

\begin{openproblem}[internal Hadamards]\label{op:hadamard}
Does $t\ge\delta(U)$ hold for every measurement-free Clifford$+T$ circuit with clean ancillas that
implements a diagonal level-three gate $U$? The supplementary notes~\cite{IsofloorSupp} prove it with
at most two internal Hadamards, when the correction is Clifford, at maximal output rank, and for
temporary-AND compilations under conditions on how the AND values are used.
\end{openproblem}

\begin{openproblem}[where the floor is exact]\label{op:exact}
Which non-Clifford pure-cubic gates have $\delta=2\nu+1$? The floor is exact on every pure-cubic gate
with $n\le5$, on layers whose blocks share single qubits along a forest, on cubics that factor, but not on three blocks that
meet pairwise in different qubits (Remark~\ref{rem:ledgerfails}) or on two of the five classes of
cubic forms at $n=6$~\cite{IsofloorSupp}.
\end{openproblem}

\begin{openproblem}[the missing unit on controlled unitaries]\label{op:controlled}
On a controlled diagonal Clifford that is not pure-cubic and has $\mu$ odd, the floor is one below
Campbell and Howard's exact count (Proposition~\ref{prop:controlled}). Is there an efficiently
computable bound that recovers the parity step of their converse~\cite[App.~C]{CampbellHowardPRA} on
every diagonal level-three gate?
\end{openproblem}

For field multiplication the ranks of the individual slices sharpen the floor $6k+1$, and they
also bound the number of Toffoli gates in Hadamard-free circuits, a bound that survives measurement
and feedforward. These results will appear separately.

\subsection*{Organization of the appendices}

The sections above can be read on their own. Appendix~\ref{sec:prelim} collects notation and
standard facts, with a table of symbols. Appendix~\ref{sec:lb} proves the floor and its ceiling,
and Appendix~\ref{app:exact} proves the exact counts for Toffoli layers and compares the floor with
Campbell and Howard's count on controlled unitaries. Appendix~\ref{app:applications} treats field
multiplication and the benchmarks.

%% file: appendices/A-prelim.tex
\section{Preliminaries}\label{sec:prelim}

This appendix fixes the notation used in the rest of the paper, collected in
Table~\ref{tab:notation}. It recalls the facts about bilinear forms and Reed--Muller codes that the
proofs need, with citations in place of proofs. It then states the Amy--Mosca correspondence in
terms of moments, the form in which the later appendices use it, and defines the cubic moment
form, its radicals and the stabilizer nullity. The gates $\CCZ$, $\CS$ and $T$ serve as running
examples.

\begin{table}[ht]
\centering
\small
\begin{tabular}{@{}l p{0.58\linewidth} l@{}}
\hline
Symbol & Meaning & Defined in\\
\hline
\multicolumn{3}{@{}l}{\small\itshape Core notation, used throughout}\\[1pt]
$n$& number of qubits of a diagonal gate $U$ & \eqref{eq:phasepoly}\\
$\phi$, $\alpha_y$ & phase polynomial $\phi(x)=\sum_y\alpha_y\,(y\cdot x)\bmod 8$ and its coefficients $\alpha_y\in\Z_8$ & \eqref{eq:phasepoly}\\
$t$ & number of odd coefficients of one representation & \S\ref{sec:phasepolys}\\
$\delta$ & least $T$-count of a Hadamard-free circuit for $U$ & \S\ref{sec:phasepolys}\\
$s=\Res(U)$ & residue: the vector $(\alpha_y\bmod 2)_{y\ne0}\in\F_2^{2^n-1}$ of a representation (the letter $s$ denotes only a residue) & \S\ref{sec:momentform}\\
$M_A$ & moment: parity of the number of odd parities $y\supseteq A$ & Definition~\ref{def:moments}\\
$M(U)$ & moment vector $(M_A)_{1\le|A|\le3}$, the same for every residue of $U$ & Fact~\ref{fact:moment}\\
$Y$, $c_a$ & matrix whose rows are the odd parities, and its column for qubit $a$ & Appendix~\ref{sec:lb}\\
$T_{abc}$, $T(w)$ & cubic moment form $M_{\{a,b,c\}}$, and its slice in direction $w$ & Definition~\ref{def:formradical}\\
$R$, $d$ & radical of $T$ and its dimension & Definition~\ref{def:formradical}\\
$Q$, $r$ & quadratic moment matrix and its rank over $\F_2$ & Definition~\ref{def:fulltensor}\\
$R^\ast$, $d^\ast$ & radical of the full moment tensor and its dimension & Definition~\ref{def:fulltensor}\\
$\mathcal A$, $\nu$ & shift group and unitary stabilizer nullity, $\nu=n-\dim\mathcal A$ & \S\ref{sec:nullitydef}\\
$T$, $S$ & overloaded symbols: $T$ is the gate $\mathrm{diag}(1,\omega)$ and, with indices or an argument, the cubic moment form, and $S$ is the phase gate and, in proofs, a set of qubits or a subspace & \\
\hline
\multicolumn{3}{@{}l}{\small\itshape Used in later appendices}\\[1pt]
$W(s)$ & ceiling: weight of the canonical residue in the coset of $s$, computed from $M(U)$ & Proposition~\ref{prop:ceiling}\\
$\eta$ & excess incidences $\eta=3m-n$ of $m$ $\CCZ$ blocks covering $n$ qubits & Lemma~\ref{lem:linearoverlap}\\
$\mu$ & least $T$-count, up to $\CCZ$ gates, of the target $U_t$ of a controlled gate & Proposition~\ref{prop:controlled}\\
$U_f$ & phase oracle $U_f\ket x=(-1)^{f(x)}\ket x$ of a cubic $f$ & Remark~\ref{rem:layernotcontrolled}\\
\hline
\end{tabular}
\caption{Notation. A subset $A\subseteq[n]$ and a vector $y\in\F_2^n$ are identified through
$y=\mathbf 1_A$. Calligraphic $\mathcal A$ is reserved for the shift group.}
\label{tab:notation}
\end{table}

\subsection{Bilinear forms and isotropic subspaces}\label{sec:bilinear}

The lower bounds of Appendix~\ref{sec:lb} come from a dimension count for subspaces on which the
dot product vanishes. We work over $\F_2=\{0,1\}$, where addition is XOR and multiplication is
AND. Write $\langle u,v\rangle=\sum_{i=1}^t u_iv_i\bmod 2$ for the dot product on $\F_2^t$ and
$V^\perp=\{u:\langle u,v\rangle=0\ \forall v\in V\}$ for a subspace $V$. The form is non-degenerate,
so $\dim V+\dim V^\perp=t$. A subspace $V$ is \emph{totally isotropic} if $V\subseteq V^\perp$. In
characteristic $2$ a vector can be orthogonal to itself, since
$\langle u,u\rangle=\wt(u)\bmod 2$ is the parity of its Hamming weight $\wt(u)$.

For example, the three columns of the matrix $Y$ in \eqref{eq:cczY}, read in $\F_2^7$, each have
weight four and every two of them share two coordinates. All their dot products vanish, so their
span is a totally isotropic subspace of dimension $3$ in $\F_2^7$.

\begin{lemma}\label{lem:iso-dim}
If $V\subseteq\F_2^t$ is totally isotropic then $\dim V\le t/2$.
\end{lemma}
\begin{proof}
$V\subseteq V^\perp$ gives $\dim V\le\dim V^\perp=t-\dim V$.
\end{proof}

The rank of the dot product restricted to a subspace can be computed from any spanning set. This
is standard~\cite{MacWilliamsSloane}.

\begin{lemma}\label{lem:gram}
Let $\mathcal C=\mathrm{span}\{c_1,\dots,c_n\}\subseteq\F_2^t$ and let $G$ be the $n\times n$ Gram matrix
$G_{ab}=\langle c_a,c_b\rangle$. Then the rank of the dot product restricted to $\mathcal C$ equals
$\rank_{\F_2}G$, and the radical $\mathcal C\cap \mathcal C^\perp$ of the restricted form has dimension
$\dim \mathcal C-\rank G$.
\end{lemma}

The radical $\mathcal C\cap\mathcal C^\perp$ in this lemma is a subspace of $\F_2^t$, whose
coordinates are parities. The radicals $R$ and $R^{\ast}$ of \S\ref{sec:momentform} are different
objects. They are subspaces of $\F_2^n$, whose coordinates are qubits, and they are defined by the
moment tensor.

\subsection{Linear and Reed--Muller codes}\label{sec:codes}

The $T$-count of a diagonal gate is a distance to a Reed--Muller code (\S\ref{sec:momentform}), and
we use only the following facts~\cite{MacWilliamsSloane}. A binary linear code of length $L$ is a
subspace $C\subseteq\F_2^L$ with dual $C^\perp$ taken with respect to the dot product. The
syndrome of a vector is its list of inner products with a basis of $C^\perp$. Two vectors have the
same syndrome exactly when they lie in the same coset $v+C$, and a \emph{coset leader} is a vector
of minimum weight in its coset. \emph{Minimum-distance decoding} asks for the codeword nearest to
a given vector $v$. If $c$ is that codeword, then $v\oplus c$ is a coset leader of $v+C$.

Two operations change the length of a code by one. Puncturing $C$ at a coordinate deletes that
coordinate from every codeword. Shortening keeps the codewords that vanish there and then deletes
it. The dual of a punctured code is the shortened dual.

The Reed--Muller code $\RM(r,m)$ consists of the evaluation vectors, over all $2^m$ points, of the
Boolean polynomials of degree at most $r$ in $m$ variables. It has dimension
$\sum_{i\le r}\binom mi$ and minimum distance $2^{m-r}$, attained by the indicators of
$(m-r)$-flats (cosets of $(m-r)$-dimensional subspaces). Its dual is $\RM(m-r-1,m)$. We write
$\RM(r,m)^*$ for $\RM(r,m)$ punctured at the all-zeros coordinate. The code that matters below is
$\RM(n-4,n)^*$, of length $2^n-1$. For $n\le3$ it is the zero code, since $\RM(r,m)=\{0\}$ for
$r<0$. At $n=4$ it is $\RM(0,4)^*=\{0,\mathbf 1\}$, whose one nonzero codeword is the all-ones
vector on the $15$ nonzero points of $\F_2^4$.

\begin{lemma}[uniform minimum distance]\label{lem:dmin}
For every $n\ge 4$, $\RM(n-4,n)^*$ has minimum distance $15$.
\end{lemma}
\begin{proof}
Puncturing lowers the weight of a codeword by at most one, so the minimum distance is at least
$2^4-1=15$. A $4$-flat through the origin has weight $16$ in $\RM(n-4,n)$ and weight $15$ after
puncturing.
\end{proof}

\subsection{Phase polynomials and the \texorpdfstring{$T$}{T}-count}\label{sec:phasepolys}

The quantity bounded throughout the paper is the $T$-count of a Hadamard-free circuit, and this
subsection defines it together with the two invariance properties used later. The circuits are
Clifford$+T$, where the Clifford group is generated by $H$, $S=\mathrm{diag}(1,i)$ and CNOT, and
$T=\mathrm{diag}(1,\omega)$ with $\omega=e^{i\pi/4}$. Identify $y\in\F_2^n$ with the set
$\{i:y_i=1\}\subseteq[n]$. The parity indexed by $y$ is the function
$x\mapsto y\cdot x=\bigoplus_{i\in y}x_i$. Every Boolean function $f:\F_2^n\to\F_2$ is uniquely a XOR
of monomials $\prod_{i\in A}x_i$, and its degree is the largest $|A|$ that occurs.

A Hadamard-free circuit on $n$ qubits built from CNOT and $T$ is diagonal and applies a sum of
$T$-rotations on parities of the input,
\begin{equation}\label{eq:phasepoly}
U\ket x=\omega^{\phi(x)}\ket x,\qquad \phi(x)=\sum_{y\in\F_2^n}\alpha_y\,(y\cdot x)\bmod 8,\quad \alpha_y\in\Z_8.
\end{equation}
We call $\phi$ the \emph{phase polynomial} and the list $(\alpha_y)_y$ a \emph{representation} of
$U$. A \emph{diagonal level-three gate} is a unitary of this form, equivalently, up to a global
phase, a diagonal gate in the third level of the Clifford hierarchy~\cite{CuiGottesmanKrishna}.
A $T$ adds $1$ to one coefficient $\alpha_y$, an $S$ adds $2$ and a $Z$ adds $4$. Two gates on the
same parity merge by adding their coefficients. Only odd coefficients require a $T$ gate, and
merging never increases the number of odd coefficients. The $T$-count of a representation is
therefore the number $t$ of odd $\alpha_y$.

Write $\delta(U)$ for the minimum of $t$ over all representations of $U$, counting $T$ and
$T^\dagger$ alike and treating diagonal Clifford gates as free. The circuit model is Hadamard-free
Clifford$+T$ in this sense. A lower bound on $\delta$ also bounds such circuits with any number of
clean ancillas (Corollary~\ref{cor:ancilla}). The standard example is $\CCZ$, whose
inclusion--exclusion representation \eqref{eq:cczid} has an odd coefficient on each of the seven
parities indexed by the nonempty subsets of $\{1,2,3\}$, so $t=7$ for it.

\begin{lemma}[Clifford bi-invariance]\label{lem:biinv}
Let $C,C'$ lie in the group generated by CNOT, $X$ and the diagonal Cliffords, and let $CUC'$ be
diagonal. Then $\delta(CUC')=\delta(U)$.
\end{lemma}
\begin{proof}
Every element of the group is $DP$ with $D$ a diagonal Clifford and $P$ an affine permutation
$x\mapsto Lx\oplus b$ built from CNOT and $X$. Writing $C=D_1P_1$ and $C'=P_2D_2$, the product
$CUC'=D_1(P_1UP_1^{-1})(P_1P_2)D_2$ is diagonal only if $P_1P_2=I$, so $CUC'=D_1(P_1UP_1^{-1})D_2$.
Conjugation by $P_1$ replaces the parity $y\cdot x$ by $(L^{-\top}y)\cdot x\oplus(L^{-\top}y)\cdot b$.
This is a bijection on nonzero parities that changes a coefficient $\alpha_y$ at most to
$-\alpha_y$ and adds a global phase, so odd coefficients stay odd. A diagonal Clifford adds only
even coefficients ($S$ adds $2$, $Z$ adds $4$, and $\CZ$ adds $2,2,6$ on $x_i$, $x_j$,
$x_i\oplus x_j$). The number of odd coefficients is therefore unchanged in every representation,
and the minimum is preserved.
\end{proof}

\begin{lemma}[restriction monotonicity]\label{lem:restrict}
Let $U$ be a diagonal level-three gate on $n$ qubits and $S\subseteq[n]$, and write $U|_S$ for the
gate induced by fixing the inputs outside $S$ to $0$. Then $\delta(U)\ge\delta(U|_S)$.
\end{lemma}
\begin{proof}
Take any representation of $U$, a multiset of parities $y$ with coefficients $\alpha_y$. Setting
$x_j=0$ for $j\notin S$ sends the parity $y\cdot x$ to $(y\cap S)\cdot x$, so the restricted
multiset represents $U|_S$. Two parities may become equal, in which case their coefficients add
and an odd pair can become even. A parity may also vanish, contributing nothing. Neither can raise
the odd count, so the restricted representation has at most as many odd parities. Minimizing over
representations gives $\delta(U|_S)\le\delta(U)$.
\end{proof}

Two further facts concern the robustness of the model, and both follow from the same count of odd
coefficients. The $T$-count is subadditive, so circuits for parts of a phase can be composed
(Lemma~\ref{lem:subadd}). Clean ancillas do not lower it when the circuit has no internal
Hadamards (Corollary~\ref{cor:ancilla}). For a phase polynomial $\phi$ we write $\delta(\phi)$ for
$\delta$ of the gate $\ket x\mapsto\omega^{\phi(x)}\ket x$.

\begin{lemma}[$\delta$ is subadditive]\label{lem:subadd}
For phase polynomials $\phi_1,\phi_2$ on the same $n$ qubits, $\delta(\phi_1+\phi_2)\le
\delta(\phi_1)+\delta(\phi_2)$.
\end{lemma}
\begin{proof}
Take representations of $\phi_1$ and $\phi_2$ with $\delta(\phi_1)$ and $\delta(\phi_2)$ odd
coefficients. Adding their coefficients parity by parity gives a representation of
$\phi_1+\phi_2$. A coefficient of the sum is odd only if one of the two summands is odd, so the sum
has at most $\delta(\phi_1)+\delta(\phi_2)$ odd coefficients.
\end{proof}

\begin{corollary}[clean ancillas without internal Hadamards]\label{cor:ancilla}
Let $K$ be a $\{$CNOT$,T\}$ circuit on $n+k$ qubits (no internal Hadamards) with
$K\ket x\ket{0^k}=\omega^{\phi(x)}\ket x\ket{0^k}$ for all $x$, implementing the diagonal
level-three gate $U$. Then its $T$-count satisfies $t\ge\delta(U)$.
\end{corollary}
\begin{proof}
Write the input as $z=(x,c)\in\F_2^{n+k}$. At every point of $K$ each wire carries a parity $y\cdot z$
of the input, and a $T$ or $T^\dagger$ on that wire contributes $\pm1$ to the coefficient of $y$.
Hence $K\ket z=\omega^{\Phi(z)}\ket{Lz}$, where $L\in GL(n+k,\F_2)$ is the permutation applied by the
CNOT gates and $\Phi(z)=\sum_y\alpha_y\,(y\cdot z)\bmod 8$ is a phase polynomial
\eqref{eq:phasepoly} on $n+k$ qubits with at most $t$ odd coefficients. The hypothesis on $K$ gives
$L(x,0)=(x,0)$ and $\Phi(x,0)\equiv\phi(x)\pmod 8$ for every $x$. Thus $L(x,c)=(x+L_{12}c,\,L_{22}c)$
for some matrices $L_{12},L_{22}$, and $\phi$ is the phase polynomial of $K$ restricted to the
ancillas in $0$. Setting $c=0$ sends each parity $y\cdot z$ to $(y\cap[n])\cdot x$, as in the proof of
Lemma~\ref{lem:restrict}. Parities that coincide merge and parities that vanish drop out, so the
restricted list is a representation of $U$ with at most $t$ odd coefficients. Hence
$\delta(U)\le t$.
\end{proof}

\subsection{The Amy--Mosca correspondence and the moment reformulation}\label{sec:momentform}

All bounds in the later appendices are statements about a few low-order parities of the odd
parity set, and this subsection shows why those parities determine the gate. Reducing the
coefficients in \eqref{eq:phasepoly} modulo $2$ gives the \emph{residue}
$s=\Res(U)=(\alpha_y\bmod 2)_y\in\F_2^{2^n-1}$, indexed by nonzero parities, which records the
parities that carry an odd phase. Amy and Mosca~\cite{AmyMosca} proved that two representations
implement the same gate up to a diagonal Clifford if and only if their residues differ by a
codeword of $\RM(n-4,n)^*$. Hence
\begin{equation}\label{eq:am}
\delta(U)=\operatorname{dist}\bigl(\Res(U),\RM(n-4,n)^*\bigr)
=\min_{c\in\RM(n-4,n)^*}\wt\bigl(\Res(U)\oplus c\bigr).
\end{equation}

The degree $n-4$ has a short explanation. Expand each parity over $\Z$ by inclusion--exclusion,
\[
y\cdot x=\sum_{\varnothing\ne A\subseteq y}(-2)^{|A|-1}\prod_{i\in A}x_i .
\]
Modulo $8$ the terms with $|A|\ge4$ vanish, so a parity contributes to the monomial
$\prod_{i\in A}x_i$ only for $|A|\le3$, with weight $1$, $-2$ or $4$. A set $\mathcal S$ of
parities, each with coefficient $1$, therefore gives the phase
$\sum_{1\le|A|\le3}(-2)^{|A|-1}N_A\prod_{i\in A}x_i$, where $N_A$ counts the parities in
$\mathcal S$ that contain $A$. This phase is a diagonal Clifford, a sum of multiples of $2x_i$ and of
$4x_ix_j$, exactly when every $N_A$ with $1\le|A|\le3$ is even. In other words, the indicator of
$\mathcal S$ must be orthogonal to every monomial $\prod_{i\in A}y_i$ with $1\le|A|\le3$. These
monomials span $\RM(3,n)$ shortened at $y=0$, whose dual is $\RM(n-4,n)^*$. For $\CCZ$ the cubic
monomial is the one that fails. Among the seven parities of \eqref{eq:cczid} only
$x_1\oplus x_2\oplus x_3$ contains $\{1,2,3\}$, while every singleton and every pair lies in an
even number of them.

\begin{definition}[moments]\label{def:moments}
For a set $\mathcal S\subseteq\F_2^n\setminus\{0\}$ of parities and $A\subseteq[n]$, the moment is
\[
M_A(\mathcal S)=\bigl|\{y\in \mathcal S:A\subseteq y\}\bigr|\bmod 2.
\]
For a vector $s\in\F_2^{2^n-1}$, $M_A(s)$ is the moment of its support.
\end{definition}

The moment $M_A$ is the syndrome coordinate of $\mathcal S$ against the monomial
$\prod_{i\in A}y_i$, so the derivation above gives the following operational form of
\eqref{eq:am}.

\begin{fact}[moment form of Amy--Mosca]\label{fact:moment}
For $1\le|A|\le3$ the moment $M_A(s)$ takes the same value $M_A(U)$ on every residue $s$ of $U$.
Two residues lie in the same coset of $\RM(n-4,n)^*$ if and only if these moments agree, and
\[
\delta(U)=\min\bigl\{\,|\mathcal S|:M_A(\mathcal S)=M_A(U)\ \text{for all }1\le|A|\le 3\,\bigr\}.
\]
The moment $M_\varnothing$ and the moments with $|A|\ge 4$ are unconstrained.
\end{fact}

We write $M_A$ for $M_A(U)$ when the gate is clear. In the language of Heyfron and
Campbell~\cite{HeyfronCampbell} the moments with $|A|\le3$ are the entries of the signature
tensor, and $\delta(U)$ is its symmetric rank.

The isotropy floor of Appendix~\ref{sec:lb} is computed from two objects built out of these
moments. The first uses only the order-three moments.

\begin{definition}[pure-cubic gate, cubic moment form, radical]\label{def:formradical}
A gate is \emph{pure-cubic} if $M_A=0$ for all $1\le|A|\le 2$. The \emph{cubic moment form} of a
gate is the symmetric tensor with $T_{abc}=M_{\{a,b,c\}}$ for distinct $a,b,c$ and $T_{abc}=0$
otherwise. For $w\in\F_2^n$ the matrix $T(w)_{qr}=\sum_aw_aT_{aqr}$ is the \emph{slice} of $T$ in
direction $w$. The \emph{radical} of $T$ is the set of directions whose slice vanishes,
\[
R=\Bigl\{w\in\F_2^n:\textstyle\sum_a w_aT_{aqr}=0\ \forall q,r\Bigr\},\qquad d=\dim R,
\]
and $T$ is \emph{non-degenerate} when $d=0$.
\end{definition}

For $\CCZ$ every singleton and every pair lies in an even number of the seven parities of
\eqref{eq:cczid}, so $\CCZ$ is pure-cubic, and $T_{123}=1$ is the only nonzero entry up to symmetry.
The slice $T(w)$ has entries $w_3$, $w_2$, $w_1$ in positions $\{1,2\}$, $\{1,3\}$, $\{2,3\}$ and
vanishes only at $w=0$, so $R=\{0\}$ and $d=0$ (\S\ref{sec:warmup}). The same gate on four qubits,
with the fourth idle, has $T(e_4)=0$, so $R=\langle e_4\rangle$ and $d=1$.

The second object keeps the moments of orders one and two as well. They are placed in a matrix, and
at the coincident indices of a three-index tensor.

\begin{definition}[quadratic moment matrix, full moment tensor]\label{def:fulltensor}
The \emph{quadratic moment matrix} of a gate is $Q\in\F_2^{n\times n}$ with $Q_{ab}=M_{\{a,b\}}$ for
$a\ne b$ and $Q_{aa}=M_{\{a\}}$, of rank $r=\rank_{\F_2}Q$. The gate is a \emph{pure quadratic} if
$M_A=0$ for all $|A|\in\{1,3\}$. The \emph{full moment tensor} has entries $M_{\{a,p,q\}}$ for all
$a,p,q\in[n]$, where coincident indices give moments of order one or two. Its radical is
\[
R^{\ast}=\Bigl\{w\in\F_2^n:\textstyle\sum_a w_a\,M_{\{a,p,q\}}=0\ \forall p,q\Bigr\},\qquad d^{\ast}=\dim R^{\ast}.
\]
\end{definition}

For $\CS$ the identity $2x_1x_2=x_1+x_2-(x_1\oplus x_2)$ gives the residue on the parities
$\{1\}$, $\{2\}$ and $\{1,2\}$. Each singleton lies in two of them and the pair in one, so
$M_{\{1\}}=M_{\{2\}}=0$ and $M_{\{1,2\}}=1$. The gate is a pure quadratic with
\[
Q=\begin{pmatrix}0&1\\1&0\end{pmatrix},\qquad r=2 .
\]
There are no triples on two qubits, so $T=0$ and $d=2$. The full tensor sees $Q$ through its
coincident indices. Taking $p=q=1$ and then $p=q=2$ in the definition of $R^{\ast}$ forces $w_2=0$
and $w_1=0$, so $R^{\ast}=\{0\}$ and $d^{\ast}=0$. The two radicals can therefore differ by the whole
space.

Given the moments, computing $d$, $d^{\ast}$ and $r$ is Gaussian elimination on at most $n^2$
rows of length $n$. This takes $O(n^3)$ operations on $n$-bit words, or $O(n^4)$ bit operations.
On a pure-cubic gate $Q=0$, $r=0$ and $R^\ast=R$.

The pure-cubic locus is invariant under \CNOT{} relabelling, which the exact counts of
Appendix~\ref{app:exact} use.

\begin{remark}[$2$-adic grading]\label{rem:grading}
The phase of a diagonal level-three gate can be written $\phi\equiv4f_3+2f_2+f_1\pmod 8$, where each
$f_j$ is an $\F_2$-polynomial of degree at most $j$ evaluated in $\{0,1\}\subset\Z$, and the carry
identity $u+v=(u\oplus v)+2uv$ brings any phase polynomial to this form. The cubic and quadratic
parts are $f_3$ and $f_2$. A CNOT relabelling replaces each variable by an $\F_2$ sum, which differs
from the integer sum by a carry. Since $4z\bmod 8$ depends only on $z\bmod 2$, carries do not affect
the cubic part, so the pure-cubic locus is closed under $GL(n,\F_2)$ and $\delta$ is constant on each
orbit. Since $2u+2v=2(u\oplus v)+4uv$, a relabelling converts quadratic terms into cubic ones, so the
pure-quadratic locus is not closed under $GL(n,\F_2)$, and we do not expect a bound depending on $r$ alone to be
tight on the quadratic side.
\end{remark}

\subsection{Stabilizer nullity and the shift group}\label{sec:nullitydef}

Stabilizer nullity is the standard magic monotone against which the isotropy floor is compared in
Appendix~\ref{sec:lb}. For $a\in\F_2^n$ write $X^a=\bigotimes_iX^{a_i}$. The \emph{shift group}
of a diagonal level-three gate $U$ is
\[
\mathcal A=\bigl\{a\in\F_2^n:\ X^aUX^a=\lambda\,UP\ \text{for some phase }\lambda\text{ and diagonal Pauli }P\bigr\}.
\]
Since $X^aUX^a\ket x=\omega^{\phi(x\oplus a)}\ket x$, a direction $a$ lies in $\mathcal A$ if and
only if
\[
\phi(x\oplus a)-\phi(x)\equiv c+4\ell(x)\pmod 8\qquad\text{for all }x,
\]
for some constant $c\in\Z_8$ and some affine function $\ell:\F_2^n\to\F_2$. Conjugating a diagonal
Pauli by $X^b$ gives a diagonal Pauli up to sign, so $\mathcal A$ is a subgroup of $\F_2^n$. For a
diagonal gate the unitary stabilizer nullity of Jiang and Wang~\cite{JiangWang} (see
also~\cite{BCHK}) is
\[
\nu=n-\dim\mathcal A=n-\log_2|\mathcal A|,
\]
and we take this as its definition. Proposition~\ref{prop:nullitygen} identifies $\mathcal A$
with $R^\ast$.

For $\CCZ$ the phase is $\phi=4f$ with $f=x_1x_2x_3$. Since $4z\bmod 8$ depends only on $z\bmod2$,
the difference $\phi(x\oplus a)-\phi(x)$ is $4$ times the derivative $f(x\oplus a)\oplus f(x)$, whose
quadratic part $a_1x_2x_3\oplus a_2x_1x_3\oplus a_3x_1x_2$ is nonzero for every $a\ne0$. So
$\mathcal A=\{0\}$ and $\nu(\CCZ)=3$. For $T$ the difference $\phi(x\oplus1)-\phi(x)=1-2x$ is not
of the form $c+4\ell(x)$, so $\nu(T)=1$. For $\CS$ the phase is $\phi=2x_1x_2$. The shifts $e_1$,
$e_2$ and $e_1\oplus e_2$ give the differences $2x_2-4x_1x_2$, $2x_1-4x_1x_2$ and $2-2x_1-2x_2$.
None has the form $c+4\ell(x)$, so $\nu(\CS)=2$.

%% file: appendices/B-floor.tex
\section{The isotropy floor}\label{sec:lb}

This appendix proves the bounds summarized in Figure~\ref{fig:pipeline}, whose single-gate case is
worked in \S\ref{sec:warmup}. The floor is proved first for pure-cubic gates (\S\ref{sec:purecubic})
and then for every diagonal level-three gate (\S\ref{sec:general}), and the radical it depends on
is identified with the stabilizer nullity (\S\ref{sec:nullityradical}). Rigidity below weight eight
(\S\ref{sec:rigid}) and a ceiling computable in $O(n^3)$ (\S\ref{sec:ceiling}) follow.

The notation is that of Appendix~\ref{sec:prelim}. For a diagonal level-three gate $U$ we write
$M_A$ for the moment of any residue of $U$, which by Fact~\ref{fact:moment} depends only on $U$.
Fix a representation of $U$ with odd parities $y_1,\dots,y_t$. Let $Y\in\F_2^{t\times n}$ be the
matrix whose rows are the $y_j$, and let $c_a\in\F_2^t$ be its $a$-th column, the indicator of
which parities contain qubit $a$. The key identities are
\begin{equation}\label{eq:cols}
\langle c_a,c_b\rangle=\sum_j (y_j)_a(y_j)_b=M_{\{a,b\}},\qquad
\langle c_a,c_a\rangle=\sum_j (y_j)_a=M_{\{a\}},
\end{equation}
and, for the triple products with $a,b,q$ distinct,
\begin{equation}\label{eq:triple}
\sum_j (c_a)_j(c_b)_j(c_q)_j=M_{\{a,b,q\}},
\end{equation}
which equals $T_{abq}$ on the pure-cubic locus.

\subsection{The pure-cubic case}\label{sec:purecubic}

This subsection is the warm-up of \S\ref{sec:warmup} with $\CCZ$ replaced by an arbitrary
pure-cubic gate. The steps are the same. The columns of $Y$ span a self-orthogonal code
(Lemma~\ref{lem:isocol}), the kernel of $Y$ lies in the radical (Lemma~\ref{lem:rankstep}), and the
all-ones vector supplies the extra one (Lemma~\ref{lem:allones}). The lower bound compares the
dimension of the column code with the radical of the cubic moment form $T$.

\begin{lemma}[isotropic columns]\label{lem:isocol}
For a pure-cubic gate, every representation has $t\ge 2\rank Y$.
\end{lemma}
\begin{proof}
By \eqref{eq:cols}, a pure-cubic gate has $\langle c_a,c_b\rangle=M_{\{a,b\}}=0$ for $a\ne b$ and
$\langle c_a,c_a\rangle=M_{\{a\}}=0$. Hence $\mathcal C=\mathrm{span}\{c_a\}$ is totally isotropic, and
Lemma~\ref{lem:iso-dim} gives $\dim \mathcal C\le t/2$. Since column rank equals row rank,
$\rank Y=\dim \mathcal C\le t/2$.
\end{proof}

The moments bound the rank of $Y$ from below.

\begin{lemma}[rank versus radical]\label{lem:rankstep}
Let $y_1,\dots,y_t$ be the odd parities of a representation of a diagonal level-three gate $U$, with
matrix $Y$ and columns $c_a$ as above. Then $\ker Y\subseteq R^{\ast}$, and hence
$\rank Y\ge n-d^{\ast}$. On a pure-cubic gate $R^{\ast}=R$, so $\rank Y\ge n-d$.
\end{lemma}
\begin{proof}
For any indices $a,p,q$, distinct or not, the sum $\sum_j(c_a)_j(c_p)_j(c_q)_j$ counts the $y_j$
that contain the set $\{a,p,q\}$, since $(c_a)_j^2=(c_a)_j$. The $y_j$ form a residue of $U$, so by
Fact~\ref{fact:moment} this count is $M_{\{a,p,q\}}$ modulo $2$. If $Yw=0$ then
$\sum_aw_ac_a=0$ in $\F_2^t$, and
\[
\sum_a w_a\,M_{\{a,p,q\}}
=\sum_j (c_p)_j(c_q)_j\Bigl(\sum_a w_ac_a\Bigr)_j=0 .
\]
This is the condition $w\in R^{\ast}$ of Definition~\ref{def:fulltensor}. Hence
$\ker Y\subseteq R^{\ast}$, and $\rank Y=n-\dim\ker Y\ge n-d^{\ast}$. On a pure-cubic gate the
moments of order one and two vanish, so the coincident-index entries of the full tensor are zero
and $R^{\ast}=R$.
\end{proof}

In the non-degenerate case a parity refinement adds one to the resulting bound, as it did for
$\CCZ$ in \S\ref{sec:warmup}.

\begin{lemma}[radical contraction of the all-ones vector]\label{lem:allones}
Let $\mathcal C=\mathrm{span}\{c_a\}\subseteq\F_2^t$ for a pure-cubic gate, and suppose the all-ones vector
$\mathbf 1\in\F_2^t$ lies in $\mathcal C$, say $\mathbf 1=\sum_{a\in S}c_a$ with $S\subseteq[n]$. Then
$\mathbf 1_S\in R$. That is, $\sum_{a\in S}T_{pqa}=0$ for all $p,q\in[n]$.
\end{lemma}
\begin{proof}
For $p=q$ every entry $T_{ppa}$ is $0$ by definition, so there is nothing to prove. Fix $p\ne q$.
Because $\mathbf 1$ is all ones, $(c_p)_j(c_q)_j=(c_p)_j(c_q)_j\mathbf 1_j$ for each coordinate
$j$. Substituting $\mathbf 1=\sum_{a\in S}c_a$ gives
\[
\sum_j (c_p)_j(c_q)_j
=\sum_j (c_p)_j(c_q)_j\Bigl(\sum_{a\in S}c_a\Bigr)_j
=\sum_{a\in S}\,\sum_j (c_p)_j(c_q)_j(c_a)_j .
\]
The left-hand side is $\langle c_p,c_q\rangle=M_{\{p,q\}}$, which is $0$ because the gate is
pure-cubic. On the right, an index $a\in S$ distinct from $p$ and $q$ contributes
$M_{\{p,q,a\}}=T_{pqa}$ by \eqref{eq:triple}. An index $a\in\{p,q\}$ contributes
$\sum_j(c_p)_j(c_q)_j=M_{\{p,q\}}=0$, which agrees with $T_{pqa}=0$ at coincident indices. Hence
$\sum_{a\in S}T_{pqa}=0$, so $\mathbf 1_S$ contracts $T$ to zero, which is the condition
$\mathbf 1_S\in R$.
\end{proof}

\begin{theorem}[cubic lower bound]\label{thm:cubicLB}
Let $U$ be a pure-cubic gate on $n$ qubits whose cubic tensor has radical dimension $d$. Then
$\delta(U)\ge 2(n-d)$, and $\delta(U)\ge 2n+1$ if $d=0$.
\end{theorem}
\begin{proof}
Fix a representation and let $\rho=\rank Y$, the dimension of the row space
$\mathrm{span}\{y_j\}$. Lemma~\ref{lem:rankstep} gives $\rho\ge n-d$, and with
Lemma~\ref{lem:isocol}, $t\ge 2\rho\ge 2(n-d)$.

For the extra one, suppose $d=0$ and, for contradiction, that some representation has $t=2n$. Then
$\dim \mathcal C=\rho=n=t/2$, so by Lemma~\ref{lem:iso-dim} the inclusion
$\mathcal C\subseteq \mathcal C^\perp$ is an equality and $\mathcal C$ is maximal isotropic. The
all-ones vector satisfies $\langle\mathbf 1,c_a\rangle=M_{\{a\}}=0$ for every $a$, so
$\mathbf 1\in \mathcal C^\perp=\mathcal C$. Write $\mathbf 1=\sum_{a\in S}c_a$ with $S$ nonempty
(nonempty because $\mathbf 1\ne 0$ for $t\ge1$). By Lemma~\ref{lem:allones},
$\mathbf 1_S\in R=\{0\}$, forcing $S=\varnothing$, a contradiction. Hence no representation has
$t=2n$, and since $t\ge2n$ already, every representation has $t\ge2n+1$. Taking the minimum,
$\delta(U)\ge 2n+1$.
\end{proof}

The extra one survives a nonzero radical. The radical directions split off as idle qubits, and the
non-degenerate case applies to what remains.

\begin{corollary}[the extra one for every non-Clifford pure-cubic gate]\label{cor:cubicplusone}
Let $U$ be a pure-cubic gate on $n$ qubits whose cubic tensor is nonzero, so that $d<n$. Then
$\delta(U)\ge 2(n-d)+1$.
\end{corollary}
\begin{proof}
Align $R$ to the last $d$ coordinates by a CNOT relabelling (an invertible $\F_2$-linear change of
the $x$-variables, which preserves $\delta$ by Lemma~\ref{lem:biinv}, preserves $d$ and transforms
the moments covariantly). Let $f$ be the cubic form with tensor $T$, so that $U$ has phase $4f$ up
to a diagonal Clifford factor. Since $T_{apq}=0$ whenever any index lies in $R$, the form $f$ has
no monomial meeting a radical coordinate. Up to a diagonal Clifford factor, which changes neither
$\delta$ nor $d$, it follows that $U=U'\otimes I_d$ with $U'$ a non-degenerate pure-cubic gate on
the first $n-d$ qubits. Then $\delta(U)=\delta(U')$. A representation of $U'$ is one of $U$, which
gives $\le$. Restricting any representation of $U$ by setting the $d$ radical inputs to $0$
(Lemma~\ref{lem:restrict}) turns each parity into a parity of the first $n-d$ variables or a
constant, without increasing the odd count, which gives $\ge$. Hence
$\delta(U)=\delta(U')\ge 2(n-d)+1$ by the non-degenerate case of Theorem~\ref{thm:cubicLB}.
\end{proof}

\begin{remark}[prior work on the pure-cubic floor]\label{rem:cubicprior}
That a pure-cubic synthesis matrix has columns of even weight and even pairwise overlap is
Campbell and Howard's Lemma~3 with vanishing quadratic matrix~\cite{CampbellHowardPRA}, and it is
also implicit in the triorthogonal codes of Bravyi and Haah~\cite{BravyiHaah}. The rank step,
Lemma~\ref{lem:rankstep}, is the standard flattening bound on the symmetric rank of Heyfron and
Campbell's signature tensor~\cite{HeyfronCampbell}, and combining it with isotropy doubles it for every
pure-cubic gate. On controlled unitaries this recovers part of~\cite[Thm.~5]{CampbellHowardPRA}. The
extra one comes from the all-ones vector (Lemma~\ref{lem:allones}), in the spirit of the parity step
of their converse~\cite[App.~C]{CampbellHowardPRA}.
\end{remark}

\subsection{The general case}\label{sec:general}

For an arbitrary gate the column Gram matrix is the quadratic moment matrix $Q$, of rank $r$, and
the columns no longer span an isotropic code. Restricting to the radical of $Q$ would lose $2r$. The
proof below loses only $r$, because the radical of the form on the column code is isotropic and has
codimension $r$ there.

\begin{theorem}[general lower bound]\label{thm:genLB}
Let $U$ be a diagonal level-three gate on $n$ qubits, with $r$ the rank of its quadratic moment
matrix and $d^{\ast}$ the dimension of the radical $R^{\ast}$ of its full moment tensor
(\S\ref{sec:momentform}). Then
\[
\delta(U)\ \ge\ 2(n-d^{\ast})-r.
\]
\end{theorem}

\begin{proof}
Fix a representation, put $\mathcal C=\mathrm{span}\{c_a\}$ and $\rho=\dim\mathcal C=\rank Y$. By
\eqref{eq:cols} the Gram matrix of the columns is $Q$, so by Lemma~\ref{lem:gram} the radical
$\mathcal C\cap\mathcal C^{\perp}$ of the dot product restricted to $\mathcal C$ has dimension
$\rho-r$. It lies in $\mathcal C^{\perp}$, which has dimension $t-\rho$ because the dot product on
$\F_2^t$ is non-degenerate. Hence $\rho-r\le t-\rho$, that is $t\ge2\rho-r$. Lemma~\ref{lem:rankstep}
gives $\rho\ge n-d^{\ast}$, and so $t\ge2(n-d^{\ast})-r$.
\end{proof}

For $\CS$ the residue $\{1\},\{2\},\{1,2\}$ has columns $c_1=(1,0,1)$ and $c_2=(0,1,1)$, so
$\rho=r=2$ and $d^{\ast}=0$. The theorem gives $t\ge2$, one below $\delta(\CS)=3$.

For a pure-cubic gate the order-$\le2$ moments vanish, so $r=0$ and $d^{\ast}=d$, and
Theorem~\ref{thm:genLB} reduces to the first part of Theorem~\ref{thm:cubicLB}. Off the pure-cubic
locus the full-tensor radical $R^{\ast}$ is needed, as the next example shows.

\begin{example}[the cubic radical alone is not enough]\label{ex:dstar}
Take the three-qubit gate with $M_{\{0,2\}}=M_{\{1,2\}}=M_{\{0,1,2\}}=1$ and all other moments $0$.
Its cubic tensor has the single entry $T_{012}=1$, so $d=0$. The quadratic matrix has rank $r=2$
and kernel $\langle e_0\oplus e_1\rangle$, and this vector also annihilates the cubic part of the
full tensor, so $R^{\ast}=\langle e_0\oplus e_1\rangle$ and $d^{\ast}=1$. At $n=3$ the code
$\RM(-1,3)^*$ is zero, so the residue is unique. Solving the moments as in
Proposition~\ref{prop:ceiling} gives the residue $\{0,1,2\},\{0,1\},\{2\}$, and $\delta=3$.
Theorem~\ref{thm:genLB} gives $2(3-1)-2=2\le3$, while substituting $d$ for $d^{\ast}$ would give
$2(3-0)-2=4>\delta$.
\end{example}

\subsection{Stabilizer nullity and the radical}\label{sec:nullityradical}

Stabilizer nullity $\nu=n-\log_2|\mathcal A|$ (\S\ref{sec:nullitydef}) is the standard efficiently
computable lower bound on $T$-count. This subsection shows that the shift group $\mathcal A$ is the
full-tensor radical $R^{\ast}$, so that the isotropy floor and $\nu$ can be compared directly.

We first compare $R^{\ast}$ with the kernel of $Q$.

\begin{lemma}[the quadratic rank is bounded by the tensor corank]\label{lem:rledn}
$R^{\ast}\subseteq\ker Q$, and hence $r\le n-d^{\ast}$.
\end{lemma}
\begin{proof}
Taking $p=q=a$ in the defining condition of $R^{\ast}$ collapses the coincident indices to
order-$\le2$ moments and gives $\sum_b w_b\,M_{\{a,b\}}=0$ for every $a$, with
$M_{\{a,a\}}=M_{\{a\}}$. That is $(Qw)_a=0$. Hence $R^{\ast}\subseteq\ker Q$ and
$d^{\ast}\le n-r$.
\end{proof}

\begin{proposition}[the shift group is the full-tensor radical]\label{prop:nullitygen}
For every diagonal level-three gate, $\mathcal A=R^{\ast}$. Hence $\nu=n-d^{\ast}$.
\end{proposition}

The proof reduces membership $w\in\mathcal A$ to a condition on a function $L_w$ modulo $4$ and
reads that condition one binary digit at a time. Modulo $2$ it says $Qw=0$, which is the
coincident-index part of $w\in R^{\ast}$. Modulo $4$, once $Qw=0$, it says that the contractions of
$w$ with the cubic moments vanish, which is the rest. The gate $\CCZ$ with $w=e_1$ is followed
through both levels.

\begin{proof}
Write $\phi(x)=\sum_y \alpha_y\,(y\cdot x)$ with $\alpha_y\in\Z_8$, reading each parity
$y\cdot x\in\{0,1\}\subset\Z$ as an integer. Since $y\cdot(x\oplus w)=y\cdot x$ when $y\cdot w=0$ and
$y\cdot(x\oplus w)=1-y\cdot x$ otherwise, we have the identity over $\Z$
\[
\phi(x\oplus w)-\phi(x)\;=\;C_w-2L_w(x),\qquad
C_w:=\!\!\sum_{y\cdot w=1}\!\!\alpha_y,\qquad
L_w(x):=\!\!\sum_{y\cdot w=1}\!\!\alpha_y\,(y\cdot x) .
\]
By \S\ref{sec:nullitydef}, $w\in \mathcal A$ iff
$\phi(x\oplus w)-\phi(x)\equiv\text{const}+4(\text{affine})\pmod 8$, which holds iff
\begin{equation}\label{eq:mod4}
L_w\;\equiv\;\text{const}+2\,(\text{affine})\pmod 4 .
\end{equation}
Split $\alpha_y=\varepsilon_y+2h_y$ with $\varepsilon_y=\alpha_y\bmod2$ the residue bit, and put
$\mathcal T_w:=\{y:y\cdot w=1,\ \varepsilon_y=1\}$. Then $L_w=S_w+2D_w$ with
$S_w(x)=\sum_{y\in \mathcal T_w}(y\cdot x)$ and $D_w(x)=\sum_{y\cdot w=1}h_y\,(y\cdot x)$. Modulo $4$ the
second term is $2(D_w\bmod2)$, and $D_w\bmod 2$ is an $\F_2$-linear function of $x$ because each
parity $y\cdot x$ is, so the even part of $\alpha$ contributes only an allowed $2(\text{affine})$ and never
obstructs \eqref{eq:mod4}. For the odd part use the carry identity for $\{0,1\}$-valued $u_i$,
\[
\sum_i u_i\;\equiv\;\Bigl(\bigoplus_i u_i\Bigr)+2\!\!\sum_{i<j}\!u_iu_j \pmod 4,
\]
which gives $S_w\equiv(\sigma_w\cdot x)+2q_w\pmod 4$ with $\sigma_w:=\bigoplus_{y\in \mathcal T_w}y$ and
$q_w(x):=\sum_{y<y'\in \mathcal T_w}(y\cdot x)(y'\cdot x)\bmod 2$.

The mod-$2$ level is $Q$. The $a$-th bit $\sigma_a$ of $\sigma_w$ is
\[
\sum_{y\in \mathcal T_w}y_a=\sum_{y:\,y_a=1}\varepsilon_y\,(y\cdot w)
=\sum_b w_b\sum_{y\supseteq\{a,b\}}\varepsilon_y=\sum_b w_b\,M_{\{a,b\}}=(Qw)_a .
\]
If $\sigma_w\ne0$ then $L_w\bmod 2=\sigma_w\cdot x$ is a nonconstant function of $x$, so
\eqref{eq:mod4} fails. Hence $Qw=0$ is necessary. For $\CCZ$ and $w=e_1$ the set $\mathcal T_w$
consists of $\{1\}$, $\{1,2\}$, $\{1,3\}$ and $\{1,2,3\}$, whose sum $\sigma_w$ is $0$, as it must
be since $Q=0$.

The mod-$4$ level is the tensor. Assume $\sigma_w=0$, so $L_w\equiv2(q_w+\text{linear})$ and
\eqref{eq:mod4} holds iff $q_w$ is affine. Expanding
$(y\cdot x)(y'\cdot x)=\sum_{a,b}y_ay'_bx_ax_b$ over $\F_2$ and using $x_a^2=x_a$, the coefficient of $x_ax_b$ with $a\ne b$ in $q_w$ is
\[
\sum_{\{y,y'\}\subseteq \mathcal T_w}(y_ay'_b+y_by'_a)
=\sum_{\substack{y,y'\in \mathcal T_w\ y\ne y'}}y_ay'_b
=\sigma_a\sigma_b-\!\!\sum_{y\in \mathcal T_w}\!y_ay_b
=\sum_{y\in \mathcal T_w}y_ay_b ,
\]
using $\sigma_w=0$ in the last step. As before
$\sum_{y\in \mathcal T_w}y_ay_b=\sum_e w_e\sum_{y\supseteq\{a,b,e\}}\varepsilon_y=\sum_e w_e\,M_{\{a,b,e\}}$.
Thus $q_w$ is affine iff $\sum_e w_e\,M_{\{a,b,e\}}=0$ for all $a\ne b$. For $\CCZ$ and $w=e_1$
the coefficient of $x_2x_3$ counts the one element $\{1,2,3\}$ of $\mathcal T_w$ containing
$\{2,3\}$, so it equals $M_{\{1,2,3\}}=1$. Hence $e_1\notin\mathcal A$, in agreement with
$\nu(\CCZ)=3$.

Together with the $a=b$ conditions from the mod-$2$ level, $w\in \mathcal A$ iff
$\sum_e w_e\,M_{\{a,b,e\}}=0$ for all $a,b$, which is the condition $w\in R^{\ast}$. Hence
$\mathcal A=R^{\ast}$ and $\nu=n-\log_2|\mathcal A|=n-d^{\ast}$.
\end{proof}

In particular $\nu$ is computable from the moments $M_A$ in $O(n^3)$ time by Gaussian elimination.
On the pure-cubic locus the moments of order at most two vanish, so $R^{\ast}=R$ and $\nu=n-d$.
Corollary~\ref{cor:cubicplusone} therefore reads $\delta\ge2\nu+1$ on every non-Clifford pure-cubic
gate.

\begin{corollary}[the floor dominates stabilizer nullity]\label{cor:floordominates}
For every diagonal level-three gate, $2(n-d^{\ast})-r\ \ge\ \nu$, with equality iff $r=n-d^{\ast}$.
\end{corollary}
\begin{proof}
By Proposition~\ref{prop:nullitygen}, $\nu=n-d^{\ast}$, so
$2(n-d^{\ast})-r-\nu=(n-d^{\ast})-r$, which is $\ge0$ by Lemma~\ref{lem:rledn}.
\end{proof}

Table~\ref{tab:smallgates} collects the invariants of the running examples. The rows for $\CS$ and
Example~\ref{ex:dstar} show that the
$+1$ does not extend off the pure-cubic locus, and that there even $\delta\ge2\nu$ fails, since both
have $\nu=2$ and $\delta=3$. The general floor is $2\nu-r$ (Corollary~\ref{cor:floordominates}).

\begin{table}[ht]
\centering\small
\begin{tabular}{@{}lcccccc@{}}
\toprule
gate & $n$ & $r$ & $d^{\ast}$ & $\nu$ & floor & $\delta$\\
\midrule
$T$ & $1$ & $1$ & $0$ & $1$ & $1$ & $1$\\
$\CS$ & $2$ & $2$ & $0$ & $2$ & $2$ & $3$\\
$\CCZ$ & $3$ & $0$ & $0$ & $3$ & $7$ & $7$\\
$\CCZ\otimes T$ & $4$ & $1$ & $0$ & $4$ & $7$ & $7$\\
$\CCZ\otimes\CCZ$ & $6$ & $0$ & $0$ & $6$ & $13$ & $13$\\
Example~\ref{ex:dstar} & $3$ & $2$ & $1$ & $2$ & $2$ & $3$\\
\bottomrule
\end{tabular}
\caption{Invariants of the running examples. The floor is $2\nu+1$ on the pure-cubic rows $\CCZ$
and $\CCZ\otimes\CCZ$ (Corollary~\ref{cor:cubicplusone}) and $2(n-d^{\ast})-r$ on the others
(Theorem~\ref{thm:genLB}). The values of $\delta$ come from Lemma~\ref{lem:rigidradius} for the
first three rows, from Theorem~\ref{thm:parccz} for $\CCZ\otimes\CCZ$ and from
Example~\ref{ex:dstar} for the last row. The value for $\CCZ\otimes T$, and every other entry, was
checked by exact decoding.}
\label{tab:smallgates}
\end{table}

\begin{remark}[comparison with stabilizer nullity]\label{rem:nomonotone}
For the non-degenerate pure-cubic gates stabilizer nullity gives $\nu=n$ (for example
$\nu(\CCZ)=3$), so Theorem~\ref{thm:cubicLB} is a factor-two improvement at the same $O(n^3)$ cost.
No magic monotone under stabilizer operations with feedforward can match it, since such a monotone
also bounds adaptive circuits, in which $\CCZ$ costs four $T$ gates~\cite{Jones}. Appendix~\ref{app:exact}
gives families on which the floor is attained.
\end{remark}

\subsection{Rigidity}\label{sec:rigid}

Several exact values in this paper need only that a light residue cannot be lightened further. The
next statement is immediate from the Amy--Mosca correspondence~\cite{AmyMosca} and the minimum
distance $15$ of Lemma~\ref{lem:dmin}.

\begin{lemma}[rigidity radius]\label{lem:rigidradius}
If a diagonal level-three gate has a residue of Hamming weight $\ell\le 7$, then
$\delta=\ell$.
\end{lemma}
\begin{proof}
For $n\le3$ the code $\RM(n-4,n)^*$ is zero, every gate has a single residue, and $\delta=\ell$ for
every $\ell$. For $n\ge4$ and any nonzero codeword $c\in\RM(n-4,n)^*$, the triangle inequality and
Lemma~\ref{lem:dmin} give
\[
\wt(\Res\oplus c)\ge\wt(c)-\ell\ge 15-\ell\ge 8>\ell.
\]
The residue itself is thus the strict minimum-weight coset leader, and by \eqref{eq:am},
$\delta=\ell$. The threshold $\ell\le 7=\lfloor(15-1)/2\rfloor$ is the unique-decoding radius.
\end{proof}

For example $\delta(T)=1$, $\delta(\CS)=3$ and $\delta(\CCZ)=7$ for all $n$. Above weight seven a
weight-$15$ codeword can lower the residue weight.

\subsection{An upper bound, also in \texorpdfstring{$O(n^3)$}{O(n\textasciicircum 3)}}\label{sec:ceiling}

A lower bound alone does not say how far a given circuit is from optimal. Solving the moment
conditions from the top down gives an upper bound that is also computable in $O(n^3)$. Write
$M(U)=(M_A)_{1\le|A|\le3}$ for the moment vector of $U$. By Fact~\ref{fact:moment} it is a complete
invariant of $U$ modulo diagonal Cliffords and determines the coset of every residue $s=\Res(U)$,
so the ceiling $W(s)$ below is computed from $M(U)$ and is the same for every residue of $U$.

The moment map $\F_2^{2^n-1}\to\F_2^{\,n+\binom n2+\binom n3}$ has kernel $\RM(n-4,n)^{\ast}$.
Restricted to the parities of weight at most three it is triangular in the subset order, hence
already surjective, and Proposition~\ref{prop:ceiling} solves it explicitly. In particular some
residue, the canonical one, uses no parity of weight four or more. A residue of minimum weight may
need such parities, as for two disjoint $\CCZ$ gates, whose canonical residue has weight $14$ while
$\delta=13$ (Theorem~\ref{thm:parccz}).

\begin{proposition}[canonical residue ceiling]\label{prop:ceiling}
Let $U$ be any diagonal level-three gate with moments $M_A$, $1\le|A|\le3$. Define $z$ on the
parities of weight at most three by solving downwards,
\[
z_A=M_A\ \ (|A|=3),\qquad
z_A=M_A\oplus\!\!\bigoplus_{y\supsetneq A,\,|y|=3}\!\! z_y\ \ (|A|=2),\qquad
z_A=M_A\oplus\!\!\bigoplus_{y\supsetneq A,\,|y|\le3}\!\! z_y\ \ (|A|=1),
\]
and $z_y=0$ for $|y|\ge4$. Then $z$ is a residue for $U$, so
\[
\delta(U)\ \le\ W(s):=\wt(z),
\]
and $\delta(U)=W(s)$ whenever $W(s)\le7$.
\end{proposition}
\begin{proof}
By Fact~\ref{fact:moment} the moments of order $\le3$ determine the coset, and by the
triangularity noted above parities of weight at most three suffice. Solving downwards, at each $|A|$ the supersets of $A$ have already been
fixed, so $z_A$ is determined and by construction $M_A(z)=\bigoplus_{y\supseteq A}z_y=M_A$ for every
$A$ with $1\le|A|\le3$. Hence $z$ lies in the coset of $U$ and $\delta\le\wt(z)$, $\delta$ being the
minimum weight over that coset. Equality below $8$ is Lemma~\ref{lem:rigidradius}. The
construction takes $O(n^3)$ time and involves no decoding.
\end{proof}

On the two pure loci the ceiling has a closed form.

\begin{corollary}[the ceiling on the pure loci]\label{cor:ceilingpure}
For a pure-cubic gate with monomial set $F$ the ceiling is
$W(s)=w(F):=\#\{A: 1\le|A|\le3,\ \#\{S\in F:A\subseteq S\}\text{ odd}\}$. For a pure quadratic with
moment graph $E$ it is $e+\mathrm{odd}$, the edge count plus the number of odd-degree vertices.
\end{corollary}
\begin{proof}
For a pure-cubic gate, $z_A=[A\in F]$ at $|A|=3$. At $|A|=2$, $M_A=0$ gives
$z_A=\#\{S\in F:A\subseteq S\}$, and at $A=\{a\}$ the weight-three parities through $a$ contribute
$\deg_F(a)$ while the weight-two parities contribute
$\sum_{b\ne a}\#\{S\in F:\{a,b\}\subseteq S\}=2\deg_F(a)\equiv0$, so $z_{\{a\}}=\deg_F(a)$ and the
total weight is $w(F)$. For a pure quadratic the weight-three level is empty, $z_{\{a,b\}}=E_{ab}$
and $z_{\{a\}}=\deg_E(a)$, giving $e+\mathrm{odd}$.
\end{proof}

The cubic and quadratic formulas are two cases of one construction, and unlike $w(F)$ the ceiling
$W(s)$ applies to mixed-degree phase polynomials. Theorem~\ref{thm:genLB}, with the extra one of Corollary~\ref{cor:cubicplusone} on the pure-cubic
locus, and Proposition~\ref{prop:ceiling} therefore bracket every gate between two quantities computed
from the moments in $O(n^3)$.

A phase polynomial with $t$ terms yields the moments in $O(tn^3)$ time. When the two ends of the
bracket coincide they determine $\delta$.

%% file: appendices/C-exact.tex
\section{Exact \texorpdfstring{$T$}{T}-counts of Toffoli layers}\label{app:exact}

On a layer of $m$ disjoint $\CCZ$ gates the isotropy floor of Appendix~\ref{sec:lb} is attained,
and $\delta=6m+1$. Every proof below has this shape, with the floor on one side and an explicit
representation that meets it on the other. Layers whose blocks share a common pair are exact too, as
are layers whose shared qubits can be routed along a forest of maximum degree three. Outside that
class the floor only brackets $\delta$. The appendix ends by comparing the floor with Campbell and
Howard's exact count on controlled unitaries (\S\ref{sec:controlled}).

\subsection{Toffoli layers}\label{app:layers}

A Toffoli layer is, up to Hadamards on the targets, a product of $\CCZ$ gates on triples of qubits,
the blocks of the layer. We treat disjoint blocks first. When blocks overlap in single qubits the
floor drops by $2$ for each incidence beyond one per qubit, and the splicing construction of
\S\ref{sec:splicing} meets the lower floor whenever the overlaps route along a forest of maximum
degree three.

\subsubsection{Disjoint blocks}\label{sec:toffoli}

A Toffoli gate is a $\CCZ$ conjugated by Hadamards on the target, so a parallel Toffoli layer of
$m$ gates on disjoint qubits is, up to Clifford gates, a layer of $m$ disjoint $\CCZ$ gates on
$n=3m$ qubits. Write $L_m=\bigotimes_{i=1}^m\CCZ_i$ for this layer. At $m=1$ the seven-$T$
circuit is standard~\cite{AMMR}, and Gosset, Kliuchnikov, Mosca and Russo~\cite{GKMR} proved it
optimal among ancilla-free Clifford$+T$ circuits by exact search. In the Hadamard-free model the
value $7$ is immediate, since $\RM(-1,3)^*=\{0\}$ and the coset of $\CCZ$ is a single word of weight
$7$. Campbell and
Howard's controlled-unitary paper~\cite{CampbellHowardPRA}
shows that $m$ disjoint $\CCZ$ gates need at most $6m+1$ $T$ gates~\cite[Cor.~1, Sec.~IV.A]{CampbellHowardPRA}
and reports the value $13$ at $m=2$ as known~\cite[Ex.~IV.2]{CampbellHowardPRA}, without a proof. That
case has six qubits, and exhaustive coset decoding confirms it. Stabilizer nullity is additive
under tensor product~\cite{JiangWang,BCHK}, so it gives $3m$. In the work citing
\cite{CampbellHowardPRA,AmyMosca,HeyfronCampbell} we found no lower bound on the layer above $3m$ or
the roughly $3.71m$ of the amortized stabilizer R\'enyi entropy~\cite{AmortizedSRE}, so to our
knowledge the lower bound of Theorem~\ref{thm:parccz} for $m\ge3$ is new. Campbell and Howard's controlled-unitary theorem is
exact~\cite[Thm.~5]{CampbellHowardPRA}, but it does not apply here.

\begin{remark}[which layers are controlled unitaries]\label{rem:layernotcontrolled}
Let $f$ be a sum of $m\ge2$ cubic monomials on blocks that pairwise share at most one qubit. Then
$U_f$ is a controlled unitary in some \CNOT{} basis, up to a diagonal Clifford factor, only if all
blocks share one qubit, as in the fan-out. In particular the disjoint layer $L_m$ with $m\ge2$ is
not one. A controlled unitary with control $\ell$ and a diagonal Clifford target has cubic part
$\ell\cdot q$ modulo quadratics, which vanishes on the hyperplane $\ell=0$, and the cubic part is
well defined in every \CNOT{} basis (Remark~\ref{rem:grading}). Pick $i$ in the support of $\ell$ and
restrict $f$ to $\ell=0$ by substituting for $x_i$. A block $B$ not containing $i$ keeps its
monomial. The substitution turns a monomial $x_ix_jx_k$ into monomials that contain $x_jx_k$, and
none of them equals the monomial of $B$, since $B$ shares at most one qubit with $\{i,j,k\}$. Hence
$f$ stays cubic on $\ell=0$ unless every block contains $i$. Since this holds for every $i$ in the
support of $\ell$, all blocks share a qubit.
\end{remark}

\begin{theorem}[Toffoli-layer optimality]\label{thm:parccz}
For every $m\ge 1$, $\delta(L_m)=6m+1$, and every $\{\CNOT,T\}$ circuit for the $\CCZ$ layer with clean
ancillas uses at least $6m+1$ $T$ gates (Corollary~\ref{cor:ancilla}). Hence no Clifford$+T$
circuit for the parallel Toffoli layer whose Hadamards all lie in the two target layers uses fewer
than $6m+1$ $T$ gates.
\end{theorem}
\begin{proof}
For the lower bound, note that the layer is pure-cubic, with $T_{pqr}=1$ exactly when
$\{p,q,r\}$ is a block. Its radical is trivial: if $w\in R$, taking $q,r$ to be two qubits of a block
forces $w$ to vanish on the third, and ranging over the three pairs of each block gives $w=0$. Hence
$d=0$, and Theorem~\ref{thm:cubicLB} gives $\delta(L_m)\ge 6m+1$.

For the upper bound, Campbell and Howard's construction gives $\delta(L_m)\le6m+1$~\cite[Cor.~1]{CampbellHowardPRA}.
An explicit set of $6m+1$ parities, and hence a $\{\CNOT,T\}$ circuit, also follows from
Proposition~\ref{prop:splicetree} with $\eta=0$, which splices $m$ copies of the seven-parity $\CCZ$
circuit by case (i) of Lemma~\ref{lem:splice}.
Since the lower and upper bounds agree, $\delta(L_m)=6m+1$.
\end{proof}

\noindent Implemented separately, the $m$ gates cost $7m$. Figure~\ref{fig:layer}(a) draws the
construction at $m=2$, and Proposition~\ref{prop:splicetree} below rebuilds it by splicing.

\input{figures/fig6-layer}

\begin{corollary}[controlled-parity layers]\label{cor:cparity}
A \emph{controlled-parity} gate applies $(-1)^{x_ax_b\,\bigoplus_{i\in P}x_i}$ for a nonempty qubit set $P$
disjoint from $\{a,b\}$. A single controlled-parity gate has $\delta=7$ for every $|P|$, and a layer
of $m$ such gates on disjoint qubit sets has $\delta=6m+1$.
\end{corollary}
\begin{proof}
Let $C$ be the Clifford that, in each block, XORs the parity qubits $P_i$ onto one representative
wire. Conjugating the layer by $C$ replaces each controlled-parity by a $\CCZ$ on three wires, on
disjoint triples. Hence the layer equals $C^{-1}(\bigotimes_i\CCZ_i)C$, and since $\delta$ is Clifford
bi-invariant its $T$-count is $\delta(L_m)=6m+1$ by Theorem~\ref{thm:parccz}. For $m=1$
the block is a single $\CCZ$, giving $\delta=7$.
\end{proof}

\noindent A single controlled-parity gate has nullity $\nu=3$ for every $|P|$.

\subsubsection{Overlapping blocks}\label{sec:overlap}

The Toffoli blocks of practical oracles rarely act on disjoint qubits. Controlled layers and swap
tests share a common control among many gates, and the carry logic of a ripple-carry adder links
consecutive Toffolis through shared carries. For $m$ blocks covering $n$ qubits, write $\eta=3m-n$ for the number of \emph{excess
incidences}, the qubit incidences in excess of one per qubit. When the blocks pairwise share at
most one qubit the floor is $6m+1-2\eta$ (Lemma~\ref{lem:linearoverlap}), so each excess incidence
lowers it by $2$. When all blocks share a pair of qubits the layer collapses to a single $\CCZ$
(Remark~\ref{rem:twoshared}). Figure~\ref{fig:overlap} sets the three regimes side by side,
\[
6m+1\ \xrightarrow{\ \text{share one qubit}\ }\ 4m+3\ \xrightarrow{\ \text{share a pair}\ }\ 7 .
\]

\input{figures/fig2-overlap}

\begin{remark}[layers sharing a pair of qubits]\label{rem:twoshared}
If all $m\ge2$ blocks share a pair of qubits, the tensor is degenerate. For $m$ blocks $\{a,b,c_i\}$ through
a common pair $\{a,b\}$, the radical relations give
$R=\{w\ \text{supported on the}\ c_i\ \text{with}\ \sum_iw_{c_i}=0\}$, of dimension $n-3$. The gate
is pure-cubic, so the floor $2(n-d)+1$ of Corollary~\ref{cor:cubicplusone} equals $7$. The layer is
Clifford-equivalent to a single $\CCZ$ (Corollary~\ref{cor:cparity}), so $\delta=7$ for every $m$
and the floor is attained.
\end{remark}

\begin{lemma}[linear overlap keeps the tensor non-degenerate]\label{lem:linearoverlap}
Let $B_1,\dots,B_m$ be distinct $\CCZ$ blocks (triples) covering $n$ qubits, such that every pair
of qubits lies in at most one block (the blocks form a \emph{linear} $3$-uniform hypergraph:
pairwise intersections have size at most one). Then the union's cubic tensor is non-degenerate,
so $\delta\bigl(\prod_i\CCZ_{B_i}\bigr)\ge2n+1$. With $\eta=3m-n$ excess incidences this is the
\emph{sharing floor} $\delta\ge6m+1-2\eta$.
\end{lemma}
\begin{proof}
Let $w\in R$ and let $p$ be any qubit, lying in some block $B=\{p,q,r\}$. The pair $\{q,r\}$ is
contained in $B$ and, by linearity, in no other block, so the radical relation at $(q,r)$ is
$\sum_a w_aT_{aqr}=w_p=0$. This holds for every covered qubit, so $R=0$ and
Theorem~\ref{thm:cubicLB} applies with $d=0$.
\end{proof}

\subsubsection{Splicing along a forest}\label{sec:splicing}

The floor of Lemma~\ref{lem:linearoverlap} is not attained on every linear pattern
(Remark~\ref{rem:ledgerfails}). The construction below attains it whenever the shared qubits route
along a forest of maximum degree three. Its case (i) is the gate-synthesis construction of Campbell and
Howard's subadditivity theorem~\cite[Thm.~6]{CampbellHowardPRA} in the language of parity sets.
Case (ii) and the routing of Proposition~\ref{prop:splicetree} are our formulation, and the single
shared control of case (ii) is already covered by their fan-out count.

The commonest overlap in practice is a single control shared by every block, as in a controlled
layer. It is the simplest case of the construction, a path along which every splice is of case (ii).

\begin{theorem}[fan-out layer]\label{thm:fanout}
Let $U_m=\prod_{i=1}^m\CCZ_{0,a_i,b_i}$ be $m$ $\CCZ$ gates sharing the single control qubit $0$,
on $n=2m+1$ qubits. Then
\[
\delta(U_m)\ =\ 4m+3\qquad\text{for every }m\ge1 .
\]
\end{theorem}
\begin{proof}
For the lower bound, note that the blocks pairwise share only qubit $0$, a linear hypergraph, so
Lemma~\ref{lem:linearoverlap} gives $\delta\ge2n+1=4m+3$.

For the upper bound, take the tree $F_0$ for the shared qubit $0$ to be the path $B_1,\dots,B_m$, so $F$ has
maximum degree two, and apply Proposition~\ref{prop:splicetree} below with $\eta=m-1$. Each of the $m-1$
splices is case (ii) of Lemma~\ref{lem:splice}, so the weight is $7+4(m-1)=4m+3$, which equals the floor.
\end{proof}

\noindent Sharing one control among $m$ Toffolis thus saves $2(m-1)$ $T$-gates against the parallel
layer. By Lemma~\ref{lem:linearoverlap} and Corollary~\ref{cor:ancilla}, no Hadamard-free
$\{\CNOT,T\}$ implementation with clean ancillas saves more. Adaptive circuits and circuits with
internal Hadamards are not covered by this statement. Campbell and Howard
obtained this count and its optimality from their controlled-unitary theorem
\cite[Thm.~5]{CampbellHowardPRA}. Here it is the floor of Lemma~\ref{lem:linearoverlap}, attained
by splicing. The comparison with their theorem on general controlled unitaries is
Proposition~\ref{prop:controlled}.

Splicing builds a parity set for a union of blocks from parity sets for the parts. It deletes one
parity $y$ of the current set and adds the parities of the new block, each XORed with $y$. The
lifted parities supply the deleted phase on $y$, so the splice costs one gate less than the new
block alone. If the new block shares a qubit $c$ with the current set, the lift of $\{c\}$ can
coincide with a parity $y\cup\{c\}$ already present, and the two gates on it cancel.

For example, take the blocks $\{a,b,c\}$ and $\{c,d,e\}$ and splice at $y=\{a\}$. Delete $\{a\}$
from the seven parities of the first block, leaving $\{b\}$, $\{c\}$, $\{a,b\}$, $\{a,c\}$, $\{b,c\}$
and $\{a,b,c\}$. Adding $a$ to the seven parities of the second block gives $\{a,c\}$, $\{a,d\}$,
$\{a,e\}$, $\{a,c,d\}$, $\{a,c,e\}$, $\{a,d,e\}$ and $\{a,c,d,e\}$. The parity $\{a,c\}$ now appears
twice and drops out, so $6+7-2=11$ parities remain. This is the floor $2n+1$ at $n=5$, and
Figure~\ref{fig:merge}(b) draws the circuit.

\begin{lemma}[splicing]\label{lem:splice}
Call a set $S$ of parities \emph{pure-cubic odd} if $|S|$ is odd and $M_A(S)=0$ for $|A|\in\{1,2\}$,
and call its \emph{block set} $B(S)=\{A:|A|=3,\ M_A(S)=1\}$. The set $\CCZ(B)$, the seven nonempty
subsets of a block, is pure-cubic odd with $|\CCZ(B)|=7$ and block set $\{B\}$. Let $P,P'$ be pure-cubic
odd on supports $X=\bigcup_{\sigma\in P}\sigma$ and $X'=\bigcup_{\sigma\in P'}\sigma$ with $|X\cap X'|\le1$,
let $y\in P$ with $y\cap X'=\varnothing$, and write $P'\oplus y=\{\sigma\cup y:\sigma\in P'\}$ for the
\emph{lift} of $P'$ by $y$. Then
\[
P\star_yP'\ :=\ \bigl(P\setminus\{y\}\bigr)\ \triangle\ \bigl(P'\oplus y\bigr)
\]
is pure-cubic odd with block set $B(P)\sqcup B(P')$, and
\begin{enumerate}
\item if $X\cap X'=\varnothing$ then $|P\star_yP'|=|P|+|P'|-1$.
\item if $X\cap X'=\{c\}$, $\{c\}\in P'$ and $y\cup\{c\}\in P$, then $|P\star_yP'|=|P|+|P'|-3$.
\end{enumerate}
\end{lemma}
\begin{proof}
Since $y\cap X'=\varnothing$, each lifted parity $\sigma\cup y$ is a disjoint union. A block lies inside
the support that carries it and $|X\cap X'|\le1$, so $B(P)$ and $B(P')$ are disjoint. The lift satisfies
$M_A(P'\oplus y)=[\,A\setminus X'\subseteq y\,]\cdot M_{A\cap X'}(P')$, with
$M_{\varnothing}(P'):=|P'|\bmod2$. If $A\cap X'=\varnothing$ the lifted term is $[A\subseteq y]$, which
cancels against the deletion of $y$. This is where oddness of $|P'|$ is used. If $A\cap X'\ne\varnothing$
then the deleted $y$ contributes nothing, since $y$ misses $X'$. If moreover $|A\cap X'|\le2$, which
includes every $A$ with $A\setminus X'\ne\varnothing$, then $M_{A\cap X'}(P')=0$ because $P'$ is
pure-cubic. Otherwise $A\subseteq X'$ and the lifted term is $M_A(P')$. Hence
$M_A(P\star_yP')=M_A(P)+[A\subseteq X']\,M_A(P')$ for $1\le|A|\le3$, which gives pure-cubicity and
the block set $B(P)\sqcup B(P')$. A symmetric difference changes the size by an even number, so
$|P\star_yP'|\equiv|P|-1+|P'|\equiv1\pmod2$.

It remains to count. Distinct $\sigma\in P'$ give distinct lifts, since each $\sigma$ lies in $X'$ and
$y$ misses $X'$. A lift $\sigma\cup y$ can equal a parity of $P$ only if $\sigma\subseteq X$, that is
$\sigma\subseteq X\cap X'$. In case (i) this is impossible, so the symmetric difference is a disjoint union
of size $(|P|-1)+|P'|$. In case (ii) it forces $\sigma=\{c\}$, and the lift $\{c\}\cup y$ equals the parity
$y\cup\{c\}$ of $P$, which differs from the deleted $y$ because $c\notin y$ (as $y$ misses $X'$). That single pair
cancels, and the size is $(|P|-1)+|P'|-2$.
\end{proof}

\begin{proposition}[tree-shaped overlap attains the sharing floor]\label{prop:splicetree}
Let a layer of $m$ blocks with linear overlap have $\eta$ excess incidences. Suppose that for each
shared qubit $q$ there is a tree $F_q$ on the blocks containing $q$, with edges labelled $q$, such
that the union $F$ of these trees is a spanning forest of the block set of maximum degree three.
Splicing along $F$ by Lemma~\ref{lem:splice}, taking case (ii) on each of the $\eta$ edges of $F$
and case (i) on the remaining $m-1-\eta$ splices, yields a representation of weight
\[
7+4\eta+6(m-1-\eta)\ =\ 6m+1-2\eta.
\]
By Lemma~\ref{lem:linearoverlap} this weight is optimal, so $\delta=6m+1-2\eta$.
\end{proposition}
\begin{proof}
We first count the weight, assuming that each splice meets the hypotheses of its case. We then show
that a new block meets the placed support only in the qubit that joins it to its parent, and
finally that its parent still has a free pair of parities to splice at.

Root each tree of $F$ at any block. Since $F$ has maximum degree three, the root has at most three
children and every other block at most two. Order the blocks so that the components of $F$ come one after
another and every block follows its parent. Start from $\CCZ(B)$ for the first block and splice in
the others one at a time, taking $P$ to be the current set and $P'=\CCZ(B)$ for the new block $B$,
of support $X'=B$. Lemma~\ref{lem:splice} keeps the current set pure-cubic odd and adds $B$ to its
block set, so the final set has the moments of the layer and represents it
(Fact~\ref{fact:moment}). Case (i) adds $6$ to the weight and case (ii) adds $4$. The tree $F_q$
has one edge fewer than the number of blocks containing $q$, so $F$ has $3m-n=\eta$ edges, and there
are $m-1-\eta$ splices between components. The weight is therefore $7+4\eta+6(m-1-\eta)$, once each splice
is shown to meet the hypotheses of its case.

The first block of a component shares no qubit with the placed blocks, since blocks in different
components are disjoint, so case (i) applies with any placed $y$, which misses $B$. Now let $B$ be
the child of $K$ in $F_c$. Suppose a placed block $B'$ contains a qubit $q\ne c$ of $B$. Then
$B'\ne K$, since $K$ already meets $B$ in $c$ and the overlap is linear. The tree $F_q$ joins $B$ to
$B'$ by a path whose first edge is labelled $q$. The placed blocks of the component form a subtree
of $F$ that contains $K$ and not $B$, so $B$ also reaches $B'$ by a path whose first edge is $BK$,
labelled $c$. These are two different paths in the forest $F$, a contradiction. Hence the placed
support meets $B$ in $\{c\}$ alone, and any placed $y$ with $c\notin y$ misses $B$.

It remains to find $y$ with $c\notin y$ and $y,\ y\cup\{c\}$ both in the current set. For a placed
block $K$, let $y_K$ be the $y$ of the splice that placed it, with $y_K=\varnothing$ for the very
first block. The first block of a later component is placed by case (i), so its $y_K$ is the $y$
used there. For a block $K$ with a parent, let $c_K$ be the qubit joining them. Call a pair
$(\sigma\cup y_K,\ \sigma\cup\{c\}\cup y_K)$ with $\varnothing\ne\sigma\subseteq K\setminus\{c\}$ a
\emph{port} of $K$ for $c$. It is a valid choice of $y$ for a child of $K$ through $c$, since $y_K$
is disjoint from $K$ by the hypothesis of Lemma~\ref{lem:splice}, so $c\notin\sigma\cup y_K$. The
current set satisfies the following invariant for every placed block $K$.
\begin{quote}
($\ast$) Until the component of $K$ is complete, the current set contains $\sigma\cup y_K$ for every
$\sigma\in\CCZ(K)$, except $\sigma=\{c_K\}$ when $K$ has a parent and except the two parities of
each port of $K$ already used by a child.
\end{quote}
Indeed, a later splice deletes only its own $y$ and the parity $y\cup\{c\}$ it cancels, and the ports
used are disjoint. By the count in Lemma~\ref{lem:splice}, the lifts added at each splice differ from
every parity already present except the one cancelled, so parities emitted for different blocks stay
distinct and no splice below another block removes a parity of $K$.

By ($\ast$) it suffices that the children of $K$ can be given disjoint ports that avoid the
parity $\{c_K\}\cup y_K$. A root $K=\{a,b,c\}$ has at most three children and a block
$K=\{c_K,u,v\}$ with a parent has at most two. The table gives, for each pattern of children, the
port assigned to each child in order, written before the lift by $y_K$.
\begin{center}
\small
\begin{tabular}{llll}
\toprule
children through & first child & second child & third child\\
\midrule
\multicolumn{4}{@{}l}{root $K=\{a,b,c\}$}\\
$c,c,c$ & $(\{a\},\{a,c\})$ & $(\{b\},\{b,c\})$ & $(\{a,b\},\{a,b,c\})$\\
$c,c,a$ & $(\{b\},\{b,c\})$ & $(\{a,b\},\{a,b,c\})$ & $(\{c\},\{a,c\})$\\
$a,b,c$ & $(\{c\},\{a,c\})$ & $(\{a\},\{a,b\})$ & $(\{b\},\{b,c\})$\\
\midrule
\multicolumn{4}{@{}l}{non-root $K=\{c_K,u,v\}$}\\
$c_K,c_K$ & $(\{u\},\{u,c_K\})$ & $(\{v\},\{v,c_K\})$ & \\
$u,u$ & $(\{v\},\{u,v\})$ & $(\{c_K,v\},\{c_K,u,v\})$ & \\
$c_K,u$ & $(\{u\},\{u,c_K\})$ & $(\{v\},\{u,v\})$ & \\
$u,v$ & $(\{v\},\{u,v\})$ & $(\{c_K,u\},\{c_K,u,v\})$ & \\
\bottomrule
\end{tabular}
\end{center}
Every other pattern at a root is one of the three root rows, or a subset of one, after relabelling,
and the remaining patterns at a non-root block follow by exchanging $u$ and $v$. In every row the
ports are disjoint and none contains the parity $\{c_K\}$.
\end{proof}

\noindent The degree bound in the hypothesis is used only through the table in the proof, which
handles at most three neighbors per block. With $\eta=0$ and $F$ edgeless the proposition splices the $m$
blocks by case (i) into a representation of weight $7+6(m-1)=6m+1$, the upper bound of
Theorem~\ref{thm:parccz}. It also covers the \emph{chain layer} $B_i=\{c_{i-1},u_i,c_i\}$ on
$n=2m+1$ qubits, the $\CCZ$ layer on the carry triples of a ripple-carry adder (not the adder itself,
which reuses each carry as the next control). Each carry $c_i$ with $1\le i<m$ lies in $B_i$ and $B_{i+1}$ only, so $F$ is the path
$B_1,\dots,B_m$, and the $\CCZ$ layer on these blocks has $T$-count $4m+3$ for every $m$.
Figure~\ref{fig:splice} draws the forest for this chain.

\input{figures/fig9-splice}

The \emph{block-overlap graph} of a layer has the blocks as vertices, two blocks adjacent when they
share a qubit. For a connected component $C$ of it, write $m_C$ for its blocks, $n_C$ for its
qubits and $\eta_C=3m_C-n_C$ for its excess incidences. The edges of $F$ inside $C$ number $\eta_C$.
A connected component always has $\eta_C\ge m_C-1$, so the construction needs $\eta_C=m_C-1$ for
every $C$, and in addition a choice of the trees $F_q$ whose union has maximum degree at most
three. The global inequality $\eta\le m-1$ does not suffice, since a triangle of blocks placed
beside a disjoint fourth block satisfies it while its triangle component does not.

\begin{remark}[the sharing floor on a component with $\eta_C\ge m_C$]\label{rem:ledgerfails}
Once some component has $\eta_C\ge m_C$ no forest exists, and the floor can fail to be attained. The
smallest instance is a triangle of blocks, $\{0,1,2\},\{0,3,4\},\{1,3,5\}$ on $n=6$. The overlap is linear and the radical is
trivial, so Lemma~\ref{lem:linearoverlap} gives $6m+1-2\eta=13$, but exact Reed--Muller coset decoding
returns $\delta=15$. The canonical residue has weight $15$ and is its own coset leader, so the
ceiling of Proposition~\ref{prop:ceiling} certifies the same value from above and the bracket is
$[13,15]$ with $\delta$ at its upper end. Up to relabelling the only other linear covering pattern
at $n\le6$ with $\eta\ge m$, the four blocks $\{0,1,2\},\{0,3,4\},\{1,3,5\},\{2,4,5\}$, also has floor
$13$ and $\delta=15$, while every pattern at $n\le6$ with $\eta\le m-1$ attains the floor.
\end{remark}

When $\eta_C=m_C-1$, any choice of the trees $F_q$ inside $C$ has union a spanning tree of the
blocks of $C$, so for $m_C\le4$ the degree condition holds automatically. In general it is not
implied by $\eta_C=m_C-1$. Suppose one qubit lies in three blocks and each other qubit of those
blocks lies in one further block ($m=9$, $n=19$). Then $\eta_C=m_C-1$, but every choice of the
trees $F_q$ gives some block four neighbors in $F$.

Remark~\ref{rem:ledgerfails} and a search of small covering patterns (in the supplementary
notes~\cite{IsofloorSupp}) support the following dichotomy.
\begin{conjecture}[dichotomy for the sharing floor]\label{conj:ledger}
A linear-overlap layer attains $6m+1-2\eta$ if and only if every connected component of its
block-overlap graph satisfies $\eta_C=m_C-1$.
\end{conjecture}
\noindent The example with $m=9$ above shows that the routing of
Proposition~\ref{prop:splicetree} does not prove the ``if'' direction in general.

\subsection{Controlled unitaries}\label{sec:controlled}

Campbell and Howard determined the $T$-count of every controlled diagonal Clifford
exactly~\cite[Thm.~5]{CampbellHowardPRA}. We compare the floor with their count.

Write the target as the square of a diagonal level-three gate $U_t$, as in their theorem. The
quantity $\mu$ below counts the $T$ gates that $U_t$ needs once its $\CCZ$ content is free. It
depends only on the order-one and order-two phases of $U_t$, which are the only phases of $U_t$
that survive in $U_t^2$.

\begin{proposition}[the floor on controlled unitaries]\label{prop:controlled}
Let $U=\ketbra0\otimes I+\ketbra1\otimes U_t^2$, controlled on a control qubit $h$, with $U_t$ a
diagonal level-three gate on the other $n-1$ qubits. Then $U_t^2$ is a diagonal Clifford with
trivial global phase. Let $\mu$ be the least $T$-count of $U_t$ up to $\CCZ$ gates. Suppose $U$ is not Clifford. If $U$ is pure-cubic, its isotropy floor is $2\mu+1$.
Otherwise it is $2\mu$.
\end{proposition}
\begin{proof}
We compute the moments of $U$ and express $\mu$ through Lempel's theorem by the rank of the matrix
$\Gamma$ of order-$\le2$ moments of $U_t$. The radical $R^{\ast}$, which gives the floor, then follows
from $\Gamma$.

The normalization of $U_t^2$ is that of~\cite[Thm.~5]{CampbellHowardPRA}. Over $\Z$, the parity
$y\cdot x\in\{0,1\}$ equals $\sum_{\varnothing\ne A\subseteq y}(-2)^{|A|-1}\prod_{a\in A}x_a$, by
induction from $u\oplus v=u+v-2uv$. If $\phi_A$ is the coefficient of $\prod_{a\in A}x_a$ in a phase
polynomial $\sum_y\alpha_y\,(y\cdot x)$, then $\phi_A\equiv(-2)^{|A|-1}\sum_{y\supseteq A}\alpha_y\pmod8$
for $|A|\le3$, and the moment of the gate is $M_A=(\phi_A/2^{|A|-1})\bmod2$.

The phase of $U_t$ is $g=\sum_ig_ix_i+\sum_{i<j}g_{ij}x_ix_j+\sum_{i<j<k}g_{ijk}x_ix_jx_k$ with every
$g_{ij}$ even and every $g_{ijk}\equiv0\pmod4$, so $U_t^2$ has phase
$2\sum_i\lambda_ix_i+4\sum_{i<j}\beta_{ij}x_ix_j$ with $\lambda_i=g_i\bmod4$ and
$\beta_{ij}=g_{ij}/2\bmod2$. The phase of $U$ is therefore
\[
\phi\ =\ x_h\Bigl(2\sum_i\lambda_ix_i+4\sum_{i<j}\beta_{ij}x_ix_j\Bigr),
\]
and its only nonzero moments are $M_{\{h,i\}}=\theta_i:=\lambda_i\bmod2$ and
$M_{\{h,i,j\}}=\beta_{ij}$. Let $\Gamma$ be the symmetric $(n-1)\times(n-1)$ matrix with $\Gamma_{ii}=\theta_i$ and
$\Gamma_{ij}=\beta_{ij}$. It is the matrix of order-$\le2$ moments of $U_t$, and $\Gamma\ne0$ because $U$ is
not Clifford.

A product of $\CCZ$ gates changes only order-$3$ moments, and every order-$3$ moment vector is that of
such a product (Fact~\ref{fact:moment}). Hence $\mu$ is the least number of parities of $\F_2^{n-1}$ with
order-$\le2$ moments $\Gamma$, that is, the least number of rows of a matrix $Y$ with $Y^\top Y=\Gamma$ in the convention of
Appendix~\ref{sec:lb}, the quantity $\mu[U_t]$ of~\cite{CampbellHowardPRA}. By
Lempel's theorem~\cite{Lempel1975} this is $\rank\Gamma$ when $\Gamma$ has a nonzero diagonal entry and
$\rank\Gamma+1$ otherwise. When $U$ is pure-cubic every $\theta_i$ vanishes, so $\Gamma$ is alternating, its rank
is even and $\mu$ is odd. The floor $2\mu+1$ is then consistent with the exact count
of~\cite[Thm.~5]{CampbellHowardPRA}.

A vector $w=(w',w_h)$ lies in $R^{\ast}$ when $\sum_bw_b\,M_{\{b,p,q\}}=0$ for all $p,q$. Taking
$p,q\ne h$ gives $w_h\Gamma_{pq}=0$, so $w_h=0$. Taking $p=h$ and $q=j\ne h$ then gives $(\Gamma w')_j=0$. Taking
$p=q=h$ gives $\theta\cdot w'=0$, which follows from $\Gamma w'=0$, since $\theta\cdot x=x^\top\Gamma x=0$ for every
$x\in\ker\Gamma$. Hence $R^{\ast}=\ker\Gamma\times\{0\}$ and $n-d^{\ast}=\rank\Gamma+1$.

If $\theta\ne0$, the gate is not pure-cubic. Its quadratic moment matrix is then nonzero only in row and
column $h$, where it equals $\theta$. Hence $r=2$, and the floor is $2(\rank\Gamma+1)-2=2\rank\Gamma=2\mu$. If $\theta=0$,
the gate is pure-cubic with $R=R^{\ast}$, so by Corollary~\ref{cor:cubicplusone} its floor is $2(n-d)+1=2\rank\Gamma+3$. Then $\Gamma$ has zero
diagonal and $\mu=\rank\Gamma+1$, so the floor is $2\mu+1$.
\end{proof}

\noindent Campbell and Howard's controlled-unitary theorem~\cite[Thm.~5]{CampbellHowardPRA} gives the
exact count $2\mu$ when $\mu$ is even and
$2\mu+1$ when $\mu$ is odd. The floor is
therefore exact on every pure-cubic controlled unitary, the fan-out included, and on every
controlled unitary with $\mu$ even. It is one below the exact count otherwise. The missing unit is the parity step of their converse~\cite[App.~C]{CampbellHowardPRA}.

%% file: figures/fig6-layer.tex
\begin{figure}[t]
\centering

\newcommand{\cn}[3]{%
  \draw[black!70] (#1,#2) -- (#1,#3);
  \node[dot] at (#1,#2) {};
  \node[targ] at (#1,#3) {};}
\newcommand{\tg}[4]{\node[#1gate] at (#2,#3) {#4};}
\tikzset{
  livegate/.style={draw,fill=white,rounded corners=1pt,minimum width=3.7mm,
                   minimum height=3.2mm,inner sep=0.5pt,font=\tiny},
  gonegate/.style={draw=black!35,fill=white,rounded corners=1pt,minimum width=3.7mm,
                   minimum height=3.2mm,inner sep=0.5pt,font=\tiny,text=black!35},
  dot/.style={circle,fill,inner sep=1.0pt},
  targ/.style={circle,draw,minimum size=2.3mm,inner sep=0pt,path picture={
      \draw (path picture bounding box.west)--(path picture bounding box.east);
      \draw (path picture bounding box.north)--(path picture bounding box.south);}},
  wire/.style={black!45,line width=0.4pt},
  wlab/.style={font=\tiny,left,inner sep=1.2pt},
  span/.style={font=\tiny,text=black!60},
  par/.style={font=\tiny,text=black!60,inner sep=0.5pt},
  hdr/.style={font=\scriptsize,anchor=west}
}

\begin{tikzpicture}[x=4.4mm,y=-4.4mm]
  \node[hdr] at (-0.6,-2.0) {(a) disjoint blocks $\{a,b,c\}$ and $\{d,e,f\}$, spliced at $y=\{a\}$:
    $6+7=13$ $T$s};
  \foreach \r/\lab in {0/{a},1/{b},2/{c},3/{d},4/{e},5/{f}} {
    \draw[wire] (0.4,\r) -- (30.6,\r);
    \node[wlab] at (0.4,\r) {$\lab$};}
  \tg{gone}{1}{0}{$T$}
  \draw[black!45,line width=0.5pt] (0.62,-0.4) -- (1.38,0.4);
  \tg{live}{1}{1}{$T$}
  \tg{live}{1}{2}{$T$}
  \cn{2}{0}{1}
  \tg{live}{3}{1}{$T^{\dagger}$}
  \cn{4}{0}{1}
  \cn{5}{1}{2}
  \tg{live}{6}{2}{$T^{\dagger}$}
  \cn{7}{1}{2}
  \cn{8}{0}{2}
  \tg{live}{9}{2}{$T^{\dagger}$}
  \cn{10}{0}{2}
  \cn{11}{0}{1}
  \cn{12}{1}{2}
  \tg{live}{13}{2}{$T$}
  \cn{14}{1}{2}
  \cn{15}{0}{1}
  \cn{16}{3}{0}
  \tg{live}{17}{0}{$T$}
  \cn{18}{4}{0}
  \tg{live}{19}{0}{$T^{\dagger}$}
  \cn{20}{3}{0}
  \tg{live}{21}{0}{$T$}
  \cn{22}{5}{0}
  \tg{live}{23}{0}{$T^{\dagger}$}
  \cn{24}{3}{0}
  \tg{live}{25}{0}{$T$}
  \cn{26}{4}{0}
  \tg{live}{27}{0}{$T^{\dagger}$}
  \cn{28}{3}{0}
  \tg{live}{29}{0}{$T$}
  \cn{30}{5}{0}
  \foreach \c/\s in {17/d,19/de,21/e,23/ef,25/def,27/df,29/f}
    \node[par] at (\c,-0.75) {$\s$};
  \draw[black!45] (0.6,5.75) -- (15.4,5.75);
  \draw[black!45] (15.6,5.75) -- (30.4,5.75);
  \node[span] at (8,6.35) {block $\{a,b,c\}$ without the $T$ on $x_a$};
  \node[span] at (23,6.35) {block $\{d,e,f\}$ lifted by $a$, all on wire $a$};
\end{tikzpicture}

\vspace{2mm}

\begin{tikzpicture}[x=4.4mm,y=-4.4mm]
  \node[hdr] at (-0.6,-2.0) {(b) blocks $\{a,b,c\}$ and $\{c,d,e\}$ sharing $c$, spliced at $y=\{a\}$:
    $5+6=11$ $T$s};
  \foreach \r/\lab in {0/{a},1/{b},2/{c},3/{d},4/{e}} {
    \draw[wire] (0.4,\r) -- (26.6,\r);
    \node[wlab] at (0.4,\r) {$\lab$};}
  \tg{gone}{1}{0}{$T$}
  \draw[black!45,line width=0.5pt] (0.62,-0.4) -- (1.38,0.4);
  \tg{live}{1}{1}{$T$}
  \tg{live}{1}{2}{$T$}
  \cn{2}{0}{1}
  \tg{live}{3}{1}{$T^{\dagger}$}
  \cn{4}{0}{1}
  \cn{5}{1}{2}
  \tg{live}{6}{2}{$T^{\dagger}$}
  \cn{7}{1}{2}
  \cn{8}{0}{1}
  \cn{9}{1}{2}
  \tg{live}{10}{2}{$T$}
  \cn{11}{1}{2}
  \cn{12}{0}{1}
  \cn{13}{2}{0}
  \cn{14}{3}{0}
  \tg{live}{15}{0}{$T^{\dagger}$}
  \cn{16}{2}{0}
  \tg{live}{17}{0}{$T$}
  \cn{18}{4}{0}
  \tg{live}{19}{0}{$T^{\dagger}$}
  \cn{20}{2}{0}
  \tg{live}{21}{0}{$T$}
  \cn{22}{3}{0}
  \tg{live}{23}{0}{$T^{\dagger}$}
  \cn{24}{2}{0}
  \tg{live}{25}{0}{$T$}
  \cn{26}{4}{0}
  \foreach \c/\s in {15/cd,17/d,19/de,21/cde,23/ce,25/e}
    \node[par] at (\c,-0.75) {$\s$};
  \draw[black!45] (0.6,4.75) -- (12.4,4.75);
  \draw[black!45] (12.6,4.75) -- (26.4,4.75);
  \node[span] at (6.5,5.35) {block $\{a,b,c\}$ without $x_a$ and $x_a\oplus x_c$};
  \node[span] at (19.5,5.35) {block $\{c,d,e\}$ lifted by $a$, without $x_a\oplus x_c$};
\end{tikzpicture}

\caption{Splicing two $\CCZ$ blocks (Lemma~\ref{lem:splice}). The struck $T$ on wire $a$ is
deleted. Each $T$ of the second block acts on wire $a$ while it holds $x_a\oplus(\sigma\cdot x)$,
for the parity $\sigma$ printed above it. (a)~Disjoint blocks. The seven lifted $T$s have
coefficients summing to $1$, as for any $\CCZ$ \eqref{eq:cczid}, so together they restore the
deleted phase on $x_a$, and the pair costs $13$ rather than $14$. This is the $6m+1$ construction of
Theorem~\ref{thm:parccz} at $m=2$. (b)~Blocks sharing qubit $c$. The lift of the parity $\{c\}$
would be $x_a\oplus x_c$, which the first block already carries with the opposite sign, so both
$T$s on that parity cancel and are left out. The remaining $11$ $T$s meet the floor $2n+1$ of
Lemma~\ref{lem:linearoverlap}. Both circuits restore their wires and implement the product of the
two $\CCZ$ gates exactly.}
\label{fig:layer}\label{fig:merge}
\end{figure}
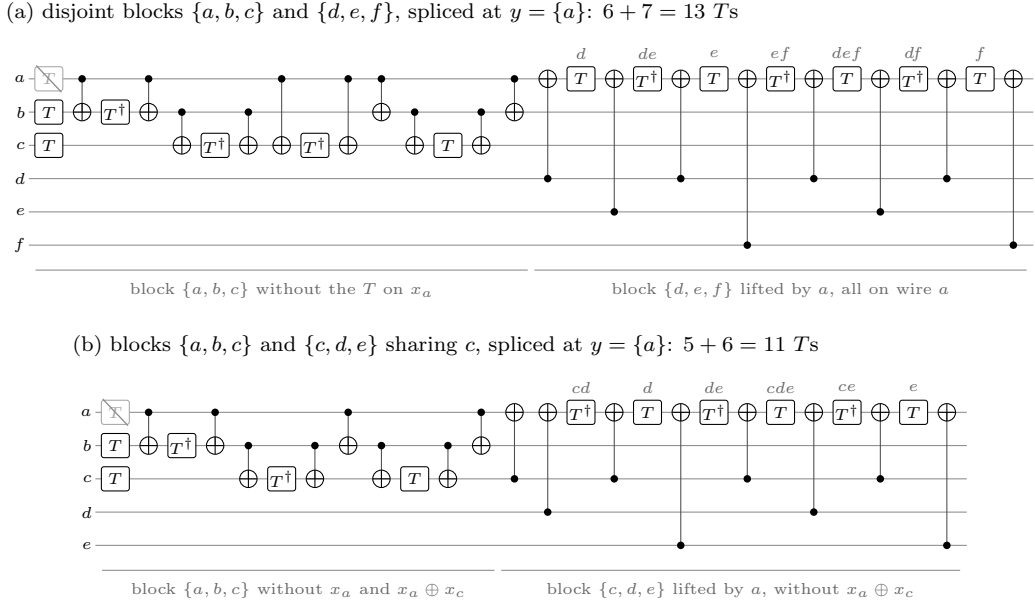

%% file: figures/fig2-overlap.tex
\begin{figure}[t]
\centering
\begin{tikzpicture}[
  >={Stealth[scale=1.0]},
  blk/.style={draw,rounded corners=2pt,minimum width=15mm,minimum height=6mm,
              font=\scriptsize,inner sep=2pt},
  pcap/.style={align=center,font=\small},
  sub/.style={align=center,font=\scriptsize,text=black!60},
  dial/.style={font=\scriptsize,text=black!55}
]

\begin{scope}[xshift=0mm]
  \node[blk] (a1) at (0,0.7)  {$\CCZ_{123}$};
  \node[blk] (a2) at (0,0)    {$\CCZ_{456}$};
  \node[blk] (a3) at (0,-0.7) {$\CCZ_{789}$};
  \node[pcap] at (0,-1.6) {disjoint};
  \node[sub] at (0,-2.15) {$d=0$,\ \ $n=3m$};
  \node[pcap] at (0,-2.9) {$\boldsymbol{6m+1}$};
\end{scope}

\begin{scope}[xshift=42mm]
  \node[blk] (b1) at (0,0.7)  {$a_1b_1$};
  \node[blk] (b2) at (0,0)    {$a_2b_2$};
  \node[font=\scriptsize] at (0,-0.5) {$\vdots$};
  \node[blk] (b3) at (0,-1.05) {$a_mb_m$};
  \node[blk,minimum width=11mm] (hh) at (-17mm,0) {control $0$};
  \foreach \b in {b1,b2,b3} \draw[black!55] (hh) -- (\b);
  \node[pcap] at (0,-1.75) {one shared control};
  \node[sub] at (0,-2.3) {$d=0$,\ \ $n=2m+1$};
  \node[pcap] at (0,-3.0) {$\boldsymbol{4m+3}$};
\end{scope}

\begin{scope}[xshift=88mm]
  \node[blk] (p) at (0,0.35) {shared pair $ab$};
  \node[blk] (c1) at (0,-0.35) {$c_1,\dots,c_m$};
  \draw[black!55] (p) -- (c1);
  \node[pcap] at (0,-1.6) {shared pair};
  \node[sub] at (0,-2.15) {$d=n-3$};
  \node[pcap] at (0,-2.9) {$\boldsymbol{7}$};
\end{scope}

\draw[->,black!50] (16mm,-2.9) -- (26mm,-2.9);
\draw[->,black!50] (60mm,-2.9) -- (70mm,-2.9);
\node[dial] at (21mm,-3.35) {more sharing};
\node[dial] at (65mm,-3.35) {more sharing};

\node[dial,align=center] at (44mm,-4.1)
  {$d=0$: the cost scales with $m$; $d=n-3$: the layer collapses to one gate};

\end{tikzpicture}
\caption{Sharing qubits interpolates the cost of a $\CCZ$ layer. Disjoint blocks give $d=0$ on $n=3m$ qubits and cost
$6m+1$ (Theorem~\ref{thm:parccz}). Fanning every block out of one shared control keeps $d=0$ but
drops $n$ to $2m+1$, so the cost falls to exactly $4m+3$ (Theorem~\ref{thm:fanout}). Each
single-qubit merge lowers the floor by two (Lemma~\ref{lem:linearoverlap}), and the $T$-count follows it
when the shared qubits route along a forest (Proposition~\ref{prop:splicetree}). Sharing a full pair
sends $d$ to $n-3$, and the layer is Clifford-equivalent to a single gate, $\delta=7$, for every $m$
(Remark~\ref{rem:twoshared}).}
\label{fig:overlap}
\end{figure}
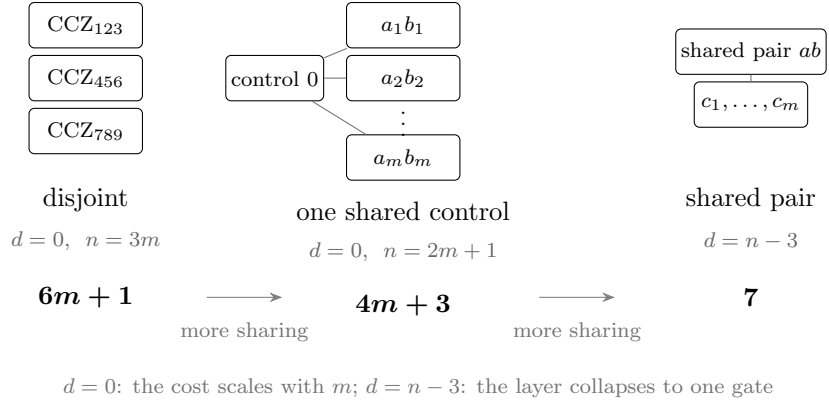

%% file: figures/fig9-splice.tex
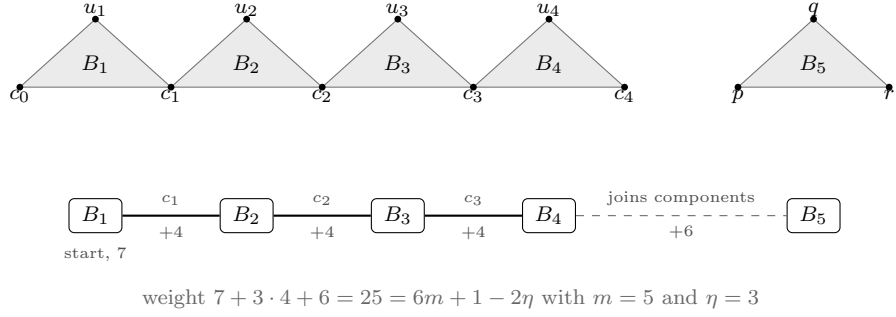
\begin{figure}[t]
\centering
\begin{tikzpicture}[x=1mm,y=1mm,
  q/.style={circle,fill,inner sep=0.9pt},
  qlab/.style={font=\scriptsize,inner sep=1pt},
  blk/.style={fill=black!8,draw=black!55,line join=round},
  bname/.style={font=\scriptsize},
  fn/.style={draw,rounded corners=2pt,font=\scriptsize,minimum width=7mm,minimum height=4.5mm,
             inner sep=1pt,fill=white},
  elab/.style={font=\tiny,text=black!70},
  note/.style={font=\scriptsize,text=black!60}
]
  \foreach \i/\x in {1/0,2/20,3/40,4/60} {
    \filldraw[blk] (\x,0) -- (\x+10,9) -- (\x+20,0) -- cycle;
    \node[bname] at (\x+10,3) {$B_{\i}$};
  }
  \filldraw[blk] (95,0) -- (105,9) -- (115,0) -- cycle;
  \node[bname] at (105,3) {$B_5$};
  \foreach \i/\x in {0/0,1/20,2/40,3/60,4/80} {
    \node[q] at (\x,0) {};
    \node[qlab,below] at (\x,0) {$c_{\i}$};
  }
  \foreach \i/\x in {1/10,2/30,3/50,4/70} {
    \node[q] at (\x,9) {};
    \node[qlab,above] at (\x,9) {$u_{\i}$};
  }
  \foreach \l/\x/\y in {p/95/0,q/105/9,r/115/0} {
    \node[q] at (\x,\y) {};
  }
  \node[qlab,below] at (95,0) {$p$};
  \node[qlab,above] at (105,9) {$q$};
  \node[qlab,below] at (115,0) {$r$};

  \foreach \i/\x in {1/10,2/30,3/50,4/70} \node[fn] (F\i) at (\x,-17) {$B_{\i}$};
  \node[fn] (F5) at (105,-17) {$B_5$};
  \foreach \a/\b/\c in {1/2/1,2/3/2,3/4/3} {
    \draw[thick] (F\a) -- node[elab,above] {$c_{\c}$} node[elab,below] {$+4$} (F\b);
  }
  \draw[dashed,black!60] (F4) -- node[elab,above] {joins components} node[elab,below] {$+6$}
    (F5);
  \node[elab,below=0.6mm of F1] {start, $7$};
  \node[note] at (57,-28) {weight $7+3\cdot4+6=25=6m+1-2\eta$ with $m=5$ and $\eta=3$};
\end{tikzpicture}
\caption{The splice forest of Proposition~\ref{prop:splicetree} on a chain layer of four
blocks $B_i=\{c_{i-1},u_i,c_i\}$, together with a disjoint fifth block. Each block is drawn as a
triangle on its three qubits. Below, the forest $F$ has one edge per shared carry, labelled by that
carry. Splicing a child along an edge of $F$ is case (ii) of Lemma~\ref{lem:splice} and adds $4$
$T$s. A block that starts a new component is spliced by case (i) and adds $6$. The first block
costs $7$.}
\label{fig:splice}
\end{figure}

%% file: appendices/D-applications.tex
\section{Applications}\label{app:applications}

This appendix applies the floor to the multiplication oracles of $\mathrm{GF}(2^k)$ and then
certifies the output of two circuit optimizers block by block.

\subsection{Field multiplication}\label{sec:fieldoracle}

Multiplication in $\mathrm{GF}(2^k)$ is the inner loop of binary-field elliptic-curve
arithmetic~\cite{BanegasBinaryEC}. Its phase oracle is a trilinear form, so the floor applies directly. Throughout this subsection $x$, $y$
and $z$ are the three $k$-qubit registers of the oracle, so here $y$ is a register and not a parity
$y_j$ of a representation.

\begin{proposition}[floor for field-multiplication oracles]\label{prop:fieldoracle}
Let $U_k\ket{x,y,z}=(-1)^{z\cdot(x\ast y)}\ket{x,y,z}$ be the multiplication phase oracle of
$\mathrm{GF}(2^k)$ on $n=3k$ qubits ($\ast$ the field product, $\cdot$ the coordinate dot).
Then $U_k$ is pure-cubic with a non-degenerate tensor, and $\delta(U_k)\ge6k+1$.
\end{proposition}
\begin{proof}
The phase is the trilinear form $z\cdot(x\ast y)$, homogeneous cubic, so the gate is pure-cubic
with $T_{x_qy_rz_l}=[\,\text{bit } l\text{ of } e_q\ast e_r\,]$ and no other entries across
registers. Let $w=(w_x,w_y,w_z)\in R$. Contracting against the pairs $(y_r,z_l)$ gives
$0=\sum_q(w_x)_q\,T_{x_qy_rz_l}=\text{bit } l\text{ of }(w_x\ast e_r)$ for all $r,l$, so
$w_x\ast e_r=0$ for every basis element. A field has no zero divisors, so $w_x=0$. Symmetrically
$w_y=0$, and contracting against $(x_q,y_r)$ gives $w_z\cdot(e_q\ast e_r)=0$ with the products
spanning $\mathrm{GF}(2^k)$ (for fixed $r$ the products $e_q\ast e_r$ form a basis, as
multiplication by $e_r\ne0$ is invertible), so $w_z=0$. Hence $d=0$ and
Theorem~\ref{thm:cubicLB} gives $\delta\ge2n+1=6k+1$.
\end{proof}

\begin{remark}[the $\mathrm{GF}(4)$ multiplier]\label{rem:gf4}
At $k=1$ the oracle is $\CCZ$ and the floor $7$ is tight. At $k=2$ the lower bounds are the
nullity $6$ and the isotropy floor $6k+1=13$. Exact Reed--Muller coset decoding shows that the merged
phase-polynomial residue of weight $17$ is its own coset leader, so $\delta(U_2)=17$ and the
term-by-term expansion is optimal once parities are merged. The unmerged schoolbook expansion has
weight $35$.

Since $\delta(U)$ is the symmetric rank of the signature tensor (Appendix~\ref{sec:prelim}), the
value $17$ was already established by AlphaTensor-Quantum~\cite{AlphaTensorQuantum}, whose Z3
search finds no symmetric decomposition of that tensor with fewer terms. The coset decoding
reproduces it independently. The value certifies published circuits optimal within a natural
class. The reversible multiplier $\ket{x,y,z}\mapsto\ket{x,y,z\oplus x{*}y}$ is $U_k$ conjugated by
Hadamards on the output register, so no multiplier circuit whose Hadamards all lie in those two
layers uses fewer than $\delta(U_k)$ $T$ gates. The best published $T$-count
for the $\mathrm{GF}(4)$ multiplier is $17$, both from a subquadratic multiplier with the
\textsc{FastTODD} optimizer~\cite{Vandaele} and from AlphaTensor-Quantum~\cite{AlphaTensorQuantum}.
\end{remark}

\begin{remark}[linear growth of field-oracle $T$-counts]\label{rem:fieldlinear}
The floor $6k+1$ is linear while the explicit circuits are not: schoolbook
multiplication uses $\Theta(k^2)$ $\CCZ$s and Karatsuba-style circuits
$\Theta(k^{1.58})$~\cite{KepleySteinwandt,vanHoof}. The $T$-count itself is nevertheless linear in
$k$, as the tensor rank shows. Write $\mu_2(k)$ for the tensor rank of
$\mathrm{GF}(2^k)$ multiplication as an $\F_2$-bilinear map. The $T$-count is within a factor $7$
of it,
\[
\mu_2(k)\ \le\ \delta(U_k)\ \le\ 7\,\mu_2(k).
\]
For the lower bound, restrict a representation's cubic tensor to index triples with one leg in each
register. The restriction is the multiplication tensor, so the number of parities is at
least its rank. For the upper bound, a rank-$\mu$ decomposition
$T=\sum_{i=1}^{\mu}a_i\otimes b_i\otimes c_i$ rewrites the phase as
$\sum_i(a_i\cdot x)(b_i\cdot y)(c_i\cdot z)$, a product of $\mu$ gates, each a $\CCZ$ on three
linear forms in different registers (a $\CCZ$ conjugated by \CNOT{} gates) contributing seven
parities. Since
$\mu_2(k)=\Theta(k)$ over $\F_2$ for every extension degree, by a Chudnovsky-type algorithm on a
descent of a Garcia--Stichtenoth tower~\cite{BalletPieltant,BilinearSurvey}, we get
$\delta(U_k)=\Theta(k)$. The floor therefore has the right order of growth. The known constants in
the linear bound are too large to compete with the explicit circuits.
\end{remark}

Sharper floors from the ranks of individual slices, and brackets against the best published
multipliers, will appear with the Toffoli-count bounds of \S\ref{sec:outlook}.

\subsection{Certifying the optimizers}\label{app:bench}

This subsection is the record behind \S\ref{sec:benchintro}. The method runs PyZX's TODD-based
phase-block optimizer (PyZX~0.10.3, \texttt{phase\_\allowbreak block\_\allowbreak optimize}) on
each circuit of the Feynman suite (commit \texttt{d2c382a2}). The run takes every circuit with
input $T$-count at most $400$, one circuit per process under a four-minute cutoff. Of the $21$
circuits left out, $18$ exceed the $T$-count limit and three (\texttt{gf2\^{}6\_mult},
\texttt{gf2\^{}7\_mult} and \texttt{qcla\_adder\_10}) reach the cutoff. The scripts, logs and
per-block records are in the repository~\cite{IsofloorCode}.
The phase-block optimizer shuffles matrix columns at random, so each circuit is run at seeds $0$ to
$4$. We report seed $4$, whose outputs have the lowest total $T$-count ($1878$, against $1880$ to
$1891$). It is also the seed with the most certified blocks. Across the five seeds the outputs have
$309$ to $311$ blocks, of which the floor certifies between $186$ and $193$ and nullity between $111$
and $114$. Among blocks with at least four $T$ gates the floor certifies $116$ to $123$ and nullity
$40$ to $44$. Table~\ref{tab:benchfull} shows eight representative circuits. The optimized circuit is partitioned with PyZX's own greedy block rule. Each
block's phase polynomial is bounded with the floor $2(n-d^{\ast})-r$, or $2\nu+1$ on the pure-cubic
locus, and exact coset decoding is added whenever a block has at most six active qubits.

A certified block has $T$-count equal to $\delta$ of its diagonal, so for this block partition and
these Hadamard positions no $\{\CNOT,T\}$ resynthesis of the block uses fewer $T$ gates. Every
exact value falls inside the bracket. None of the certified blocks is pure-cubic, so every
certificate comes from the general floor of Theorem~\ref{thm:genLB}. The block floors bound only
circuits on the same qubits with the same partition.

Two circuits show what lies outside that scope. On \texttt{csum\_mux\_9} the block floors sum to
$73$, while AlphaTensor-Quantum reports $71$ using $42$ qubits instead of
$30$~\cite{AlphaTensorQuantum}. The twelve multipliers and \texttt{mod5\_4} are the only circuits of
the suite that are mismatch-free as a whole, meaning that no Toffoli uses an earlier target as a
control. On \texttt{mod5\_4} a Hadamard on one wire turns the
circuit into a $\CNOT$ network times a pure-cubic diagonal gate with $d=2$, whose floor
$2(5-2)+1=7$ is the optimum AlphaTensor-Quantum reports from a Z3 search~\cite{AlphaTensorQuantum}.
The phase-block optimizer splits \texttt{mod5\_4} into two blocks of total $T$-count $8$, one above
this whole-circuit optimum, so its two certificates hold relative to that block partition.

In the $\mathrm{GF}(2^k)$ multipliers every Toffoli targets the output register, which is never
used as a control. The files group the Toffolis into two blocks, each conjugated by Hadamards on
the output register, with the \CNOT{} network of the modular reduction on that register between
them. Conjugated by those Hadamards the reduction network is again a \CNOT{} network, so the circuit
is a Hadamard layer, a \CNOT{} network times a diagonal gate, and a Hadamard layer. It maps
$\ket{x,y,z}\mapsto\ket{x,y,Lz\oplus x{*}y}$ for an invertible $L$, which is the multiplier up to a
\CNOT{} network on the output, so Remark~\ref{rem:gf4} applies to it. The structure and the radical are
checked for all twelve benchmark multipliers, $k=4$ to $256$~\cite{IsofloorCode}.

\begin{table}[!htbp]
\centering\footnotesize
\setlength{\tabcolsep}{3pt}
\begin{tabular}{@{}lrrrrrrrrrrr@{}}
\toprule
circuit & $n$ & $T$ in & $T$ emitted & certified & by $\nu$ & $\sum$ floor & $T$ merged & max gap & decoded & at floor & at count\\
\midrule
\texttt{adder\_8} & $24$ & $399$ & $250$ & $8/20$ & $5$ & $214$ & $246$ & $5$ & $1$ & $1$ & $1$\\
\texttt{barenco\_tof\_10} & $19$ & $224$ & $174$ & $16/33$ & $8$ & $129$ & $146$ & $1$ & $32$ & $15$ & $32$\\
\texttt{csum\_mux\_9} & $30$ & $196$ & $76$ & $0/3$ & $0$ & $73$ & $76$ & $1$ & $0$ & $0$ & $0$\\
\texttt{gf2\^{}4\_mult} & $12$ & $112$ & $50$ & $0/1$ & $0$ & $25$ & $50$ & $25$ & $0$ & $0$ & $0$\\
\texttt{grover\_5} & $9$ & $336$ & $223$ & $29/49$ & $27$ & $185$ & $205$ & $1$ & $47$ & $29$ & $46$\\
\texttt{mod5\_4} & $5$ & $28$ & $8$ & $2/2$ & $1$ & $8$ & $8$ & $0$ & $2$ & $2$ & $2$\\
\texttt{qft\_4} & $5$ & $69$ & $67$ & $34/37$ & $31$ & $64$ & $67$ & $1$ & $37$ & $34$ & $37$\\
\texttt{vbe\_adder\_3} & $10$ & $70$ & $32$ & $5/5$ & $0$ & $28$ & $28$ & $0$ & $1$ & $1$ & $1$\\
\midrule
all $24$ & & & & $193/311$ & $113$ & $1563$ & $1800$ & & $239$ & $159$ & $234$\\
\bottomrule
\end{tabular}
\caption{Eight representative circuits of the suite at seed $4$. The totals row is over all $24$
circuits that the phase-block optimizer finishes within the cutoff, and the full table is in the
supplementary notes~\cite{IsofloorSupp}. Certified counts blocks whose floor equals their count, and
by $\nu$ those that stabilizer nullity alone certifies. $T$ merged is the count after parities are
merged inside each block, and max gap the largest per-block difference between that count and the
floor. Decoded counts blocks small enough to decode exactly, at floor those whose exact count equals
the floor, and at count those whose exact count equals the optimizer's.}
\label{tab:benchfull}
\end{table}

Every pure-cubic block small enough to decode is suboptimal. There are three, a single $\CCZ$ up to $\CNOT$s in the phase-block optimizer's output for
\texttt{grover\_5} and blocks of \texttt{mod5\_4} and \texttt{csum\_mux\_9} in the output of phase
teleportation. Each carries one $T$ gate more than its floor $2\nu+1$, and decoding shows that the
floor is its exact count. Among the $239$ blocks with at least four $T$ gates (not the decoded blocks of
Table~\ref{tab:benchtight}), the floor certifies $121$ and nullity
$41$, and the largest certified block has $22$ $T$
gates. Exact decoding is affordable up to six active qubits, and the floor matters most beyond
that. Of the $72$ blocks with $7$ to $28$ active qubits it certifies $37$, against $15$ for nullity.
Phase teleportation~\cite{KissingerVdW} gives a similar
picture, with $252$ of $337$ blocks certified on $27$ circuits against $141$ for nullity ($248$ of
$318$ on the $24$ circuits both optimizers finish), $125$
against $14$ among blocks with at least four $T$ gates, and $33$ against $17$ among the $61$ blocks
beyond decoding ($7$ to $34$ active qubits), the largest with $36$ $T$ gates on $23$ active qubits. For the
$515$ blocks small enough to decode exactly, the floor was never more than two below the true count.

Since the floor is never below $\nu$ (Corollary~\ref{cor:floordominates}), every block that nullity
certifies is certified by the floor too. Exact decoding finds five suboptimal blocks, one each in
\texttt{grover\_5}, \texttt{ham15-low} and \texttt{mod\_mult\_55} and two in \texttt{hwb6}, and saves
$8$ $T$ gates on them. The two largest gaps, $51$ on \texttt{gf2\^{}5\_mult} and $25$ on
\texttt{gf2\^{}4\_mult}, are on multipliers, where the single block carries the floor $6k+1$ of
$U_k$.

Table~\ref{tab:benchtight} gives the distance from the floor to the exact $T$-count on every block
small enough to decode, for both optimizers. On the output
of the phase-block optimizer the total gap between block counts and floors is $237$ $T$ gates. It
splits as $81$ of floor slack on the $239$ decoded blocks, $8$ of excess over the optimum on the
five suboptimal ones, and $148$ on the $72$ blocks of $7$ to $28$ active qubits that decoding
cannot reach, of which $35$ carry a nonzero gap.

\begin{table}[!htbp]
\centering\small
\begin{tabular}{@{}lrrrr@{}}
\toprule
optimizer & decoded blocks & distance $0$ & distance $1$ & distance $2$\\
\midrule
phase-block (seed $4$) & $239$ & $159$ & $79$ & $1$\\
phase teleportation & $276$ & $221$ & $47$ & $8$\\
\bottomrule
\end{tabular}
\caption{The distance from the floor to the exact $T$-count, over every block with at most six active
qubits, where exact coset decoding is affordable. Distance $0$ means the floor is exact on the block.}
\label{tab:benchtight}
\end{table}

The certificate does not depend on the optimizer. Phase teleportation~\cite{KissingerVdW} changes
only the location and value of phases, so its output is bounded with the same block rule. The
per-circuit record for it, and the full version of Table~\ref{tab:benchfull}, are in the
supplementary notes~\cite{IsofloorSupp}.

%% file: sections/99-disclosure.tex
\section*{AI Disclosure}

The author conceived the isotropy floor and proved the main results of this paper: the cubic and
general lower bounds (Theorems~\ref{thm:cubicLB} and~\ref{thm:genLB}), the identity $\nu=n-d^{\ast}$
(Proposition~\ref{prop:nullitygen}), the optimality of the Toffoli layer (Theorem~\ref{thm:parccz})
and the splicing construction (Proposition~\ref{prop:splicetree}). Generative artificial intelligence
tools (Claude Opus 4.5, Claude Opus 5.5 and Claude Fable) assisted with the secondary results,
including those of the supplementary notes, with the verification code and the benchmark
experiments, and with drafting the text. The author reviewed and verified all of this assistance,
including every argument, line of code and passage of text, and takes full responsibility for the
content of this work.